\documentclass[12pt]{article}

\usepackage{amsthm,amsmath,amssymb,bbm,bm}
\usepackage{natbib}
\usepackage{multirow}
\usepackage[pdftex]{graphicx}
\usepackage{subfigure}
\usepackage{wrapfig}

\usepackage{array}
\usepackage{url}

\usepackage{algorithmic}
\usepackage[linesnumbered, ruled, vlined]{algorithm2e}

\usepackage{mathrsfs}
\usepackage{dsfont}
\usepackage{titling}

\usepackage{relsize}
\usepackage{rotating}
\usepackage{enumitem}
\usepackage{float}
\usepackage[usenames,dvipsnames,svgnames,table]{xcolor}
\usepackage{setspace}

\newtheorem{remark}{Remark}	
\newtheorem{assumption}{Assumption}

\newtheorem{theorem}{Theorem}[section]
\newtheorem{lemma}{Lemma}

\newtheorem{corollary}{Corollary}

\newtheorem{scenario}{Scenario}

{\textwidth=6in}

\newcommand{\bX}{{\boldsymbol{X}}}
\newcommand{\bx}{{\boldsymbol{x}}}

\newcommand{\bu}{{\boldsymbol{u}}}
\newcommand{\bw}{{\boldsymbol{w}}}
\newcommand{\bU}{{\boldsymbol{U}}}
\newcommand{\bW}{{\boldsymbol{W}}}

\newcommand{\btheta}{{\boldsymbol{\theta}}}

\newcommand{\tbin}{{\text{Bin}}}
\newcommand{\tunif}{{\text{Unif}}}
\newcommand{\tpoi}{{\text{Poi}}}
\newcommand{\tbeta}{{\text{Beta}}}
\newcommand{\tgamma}{{\text{Gamma}}}
\newcommand{\texp}{{\text{Exp}}}
\newcommand{\tnorm}{{\text{Normal}}}

\newcommand{\yifan}[1]{\textcolor{black}{#1}\xspace}

\begin{document}


\title{\vspace{-3cm}  {\LARGE A fiducial approach to nonparametric deconvolution problem: discrete case}}

\author{
Yifan Cui\thanks{Center for Data Science, Zhejiang University},
Jan Hannig\thanks{Department of Statistics and Operations Research, University of North Carolina at Chapel Hill}}
\date{}

\maketitle
\vspace{-2.2cm}
\begin{abstract}
Fiducial inference, as generalized by \cite{hannig2016generalized}, is applied to nonparametric $g$-modeling \citep{10.1093/biomet/asv068} in the discrete case. We propose a computationally efficient algorithm to sample from the fiducial distribution, and use generated samples to construct point estimates and confidence intervals.
We study the theoretical properties of the fiducial distribution and perform extensive simulations in various scenarios. The proposed approach gives rise to good statistical performance
in terms of the mean squared error of point estimators and coverage of confidence intervals.
Furthermore, we apply the proposed fiducial method to estimate the probability of each satellite site being malignant using gastric adenocarcinoma data with 844 patients \citep{10.1093/biomet/asv068}.
\end{abstract}

\noindent {\bf keywords}
Fiducial inference, Empirical Bayes,  Confidence intervals, Nonparametric deconvolution

\section{Introduction}\label{sec:intro}

\cite{Efron2014TwoMS,10.1093/biomet/asv068,narasimhan2016g} studied the following important problem:
An unknown distribution function $F(\theta)$ yields unobservable realizations $\Theta_1$, $\Theta_2$, \ldots , $\Theta_n$, and each $\Theta_i$ produces an observable value $X_i$ according to a known probability mechanism. 
The goal is to estimate the unknown distribution function from the observed data. In this paper, we aim to provide a generalized fiducial solution to the same problem in the case where $X_i$ given $\Theta_i$ follows a discrete distribution \yifan{that has known probability mass function $g_i$ and distribution function $G_i$.
Following \cite{10.1093/biomet/asv068}, we use the terminology deconvolution here as the marginal distribution of $X_i$ admits the following form} (see also Equation~(6) of \cite{10.1093/biomet/asv068} as well as Equation~(7) of \cite{narasimhan2016g}):
\begin{align}
\int g_i(x_i|\theta_i) dF_i(\theta_i). \label{eq:lk}
\end{align}

\cite{10.1093/biomet/asv068} proposed an empirical Bayes deconvolution approach to estimating the distribution of $\Theta$ from the observed sample $\{X_i, i=1,\ldots,n\}$, where the only requirement is a known specification of the distribution for $X_i$ given $\Theta_i$.
The empirical Bayes deconvolution since developed has seen tremendous success in many scientific applications including causal inference \citep{lee2019}, single-cell analysis \citep{wang2018gene}, cancer study \citep{gholami2015number,shen2019randomized}, clinical trials \citep{shen2018using,shen2019towards} and many other fields \citep{dulek2018empirical}.
Moreover, for a classic Bayesian data analysis, as noted in \cite[p19]{ross2018dirichletprocess} and \cite{gelman2013bayesian},
a single distribution prior may sometimes be unsuitable and hence the prior choice is dubious.
Efron's empirical Bayes deconvolution would be one of the alternatives since the obtained estimator of distribution of  $\Theta$ can be used as a prior distribution to produce posterior approximations \citep{narasimhan2016g}. 
\yifan{The discrete deconvolution problem is common in practice, e.g., the famous missing species problem is an example of Poisson deconvolution, and many studies in gene expression analysis use Poisson or negative binomial distribution \citep{robinson2010edger,siggenes}.}

Fiducial inference can be traced back to R. A. Fisher \citep{Fisher1930,Fisher1933}  who introduced the concept as a potential replacement of the Bayesian posterior distribution.
\cite{Hannig2009,hannig2016generalized} showed that fiducial distributions can be related to empirical Bayes methods, which are widely used in large-scale parallel inference problems \citep{efron2012large,efron2019,efron2019rejoinder}. \cite{Efron1998} pointed out objective Bayes theories also have connections with fiducial inference.
Other fiducial related approaches include Dempster-Shafer theory \citep{Dempster2008, EdlefsenLiuDempster2009,MartinZhangLiu2010,HannigXie2012}, inferential models \citep{martin2013inferential,MartinLiu2013b,martin2015inferential,martin2015marginal},
 confidence distributions \citep{SchwederHjort2002,XieSinghStrawderman2011, XieSingh2013,XieEtAl2013, schweder2016confidence, HjortSchweder2018,shen2019jasa} and higher order likelihood expansions and implied data-dependent priors \citep{Fraser2004,Fraser2011}. 

We propose a novel fiducial approach to modeling the distribution function $F$ of $\Theta$ nonparametrically.
In particular, we propose a computationally efficient algorithm to sample from the {\em generalized fiducial distribution} (GFD) (i.e., a distribution on a set of distribution functions), and use generated samples to construct statistical procedures.
The pointwise median of the GFD is used as point estimate, and appropriate quantiles of the GFD evaluated at a given point
provide pointwise confidence intervals.
We also study the theoretical properties of the fiducial distribution. Extensive simulations in various scenarios show that the proposed fiducial approach is a good alternative to existing methods such as Efron's $g$-modeling.
We apply the proposed fiducial approach to intestinal surgery data to estimate the probability of each satellite site being malignant for the patient, see  \cite{10.1093/biomet/asv068} for the empirical Bayes approach. The resulting fiducial estimate of \yifan{the} distribution function reflects the observed patterns of raw data.

The remainder of the article is organized as follows. In Section~\ref{sec:method}, we present the mathematical framework for the fiducial approach to nonparametric deconvolution problem. In Section~\ref{sec:theory}, we establish an asymptotic theory which verifies the frequentist validity of the proposed fiducial approach. Extensive simulation studies are presented in Section~\ref{sec:simulation}. We also illustrate our
method using intestinal surgery data in Section~\ref{sec:realdata}. The article
concludes with a discussion of future work in Section~\ref{sec:discussion}. Some needed technical results and additional simulations are
provided in the Appendix and Supplementary Material.

\section{Methodology}\label{sec:method}

\subsection{Data generating equation}\label{s:DGE}
In this section, we first explain the definition of a \yifan{GFD} and then demonstrate how to apply it to the deconvolution problem. We start by expressing the relationship between the data $X_i$ and the parameter $\Theta_i$ using
\begin{equation}\label{eq:DGE}
X_i=G_i^{-1}(U_i,\Theta_i),\quad \Theta_i = F^{-1}(W_i),\quad i=1,\ldots n,
\end{equation}
where $U_i,W_i$ are i.i.d. $\tunif(0,1)$, $G_i(\cdot,\theta_i)$ are known distribution functions of discrete random variables supported on integers, $G_i$ are defined and non-increasing in $\theta_i\in \mathcal S \subset \mathbbm R$ for all $i$, and $F$ is the unknown distribution function with the support in the interval $\mathcal S$.
We are interested in estimating the unknown distribution function $F(\theta)$.

Recall that $F^{-1}(w)=\inf\{\theta: F(\theta)\geq w\}$ \cite[p54]{CasellaBerger2002}, and $F^{-1}(w)=\theta$ if and only if $F(\theta)\geq w>F(\theta-\epsilon)$ for all $\epsilon>0$. We denote $G_i^{*}(x_i,u_i)=\sup\{\theta:G_i(x_i,\theta)\geq u_i\}$ with the usual understanding that $\sup \emptyset$ is smaller than all elements of $\mathcal S$.

\begin{remark}
\yifan{To facilitate a better understanding, we provide two examples for $G_i^{*}(x_i,u_i)$ here. If $X$ follows a binomial distribution, $G_i$ is the CDF of binomial distribution with number of trials $m_i$ and $G_i^{*}(x_i,u_i)$ is the $(1-u_i)$ quantile of $\tbeta(x_i+1,m_i-x_i)$;
if $X$ follows a Poisson distribution, $G_i$ is the CDF of Poisson distribution and $G_i^{*}(x_i,u_i)$ is the $(1-u_i)$ quantile of $\tgamma(x_i+1,1)$.}
\end{remark}

If $G_i(\cdot,\theta_i)$ is continuous in $\theta_i$, $G_i^{*}(x_i,u_i)$ is the solution (in $\theta_i$) to the equation $G_i(x_i,\theta_i)= u_i$. By Lemma~\ref{lemma:equiv} in Section~\ref{thm:0} of the Appendix, $x_i=G_i^{-1}(u_i,\theta_i)$ if and only if $\theta_i\in (G_i^{*}(x_i-1,u_i),G_i^{*}(x_i,u_i)]$. 
Combining $G_i^{*}(x_i-1,u_i)< \theta_i \leq G_i^{*}(x_i,u_i)$ and $F(\theta_i -\epsilon) <w_i\leq F(\theta_i)$ for all $\epsilon$, consequently the inverse of the data generating equation \eqref{eq:DGE} is
\begin{equation}
Q_{\bx}(\bu,\bw) =
\{F\, :\ \, F(G_i^*(x_i-1,u_i))< w_i\leq F(G_i^*(x_i,u_i)), i=1,\ldots,n\}.  \label{eq:constraint}
\end{equation}

\begin{remark}
\yifan{Note that given $\bx$, $\bu$ and $\bw$, $Q_{\bx}(\bu,\bw)$ is a set of CDFs.}
\end{remark}

By Lemma~\ref{lemma:equiv2} in Section~\ref{thm:0} of the Appendix,
$Q_{\bx}(\bu,\bw)\neq\emptyset$ if and only if $\bu, \bw$ satisfy:
\begin{equation}\label{eq:constraint2}
\mbox{whenever $G_i^*(x_i,u_i)\leq G_j^*(x_j-1,u_j)$ then $w_i < w_j$.}
\end{equation}

A \yifan{GFD} is obtained by inverting the data generating equation, and \cite{hannig2016generalized} proposed a general definition of GFD. However, in order to simplify the presentation, we use an earlier, less general version in \cite{Hannig2009}. These two definitions are equivalent for the models considered here.
Suppose $(\bU^\star,\bW^\star)$ are uniformly distributed on the set $\{(\bu^\star,\bw^\star): Q_{\bx}(\bu^\star,\bw^\star)\neq\emptyset\}.$ A GFD is then the distribution of any element of the random set 
$
 \overline{Q_{\bx}(\bU^\star,\bW^\star)},
$ 
where the closure is in the weak topology on the space of probability measures on $\mathcal S$ and the element is selected so that it is measurable, i.e., a random distribution function on $\mathcal S$.
Given observed data $(x_1,\ldots,x_n)$, we define random functions $F^U$ and $F^L$ as follows: for each $\theta\in\mathcal S$ and $(\bU^\star,\bW^\star)$,
\[
F^U(\theta) \equiv \min\{W^\star_i, \mbox{ for all $i$ such that } \theta<G_i^*(x_i-1, U^\star_i)\},
\]
and
\[
F^L(\theta) \equiv \max\{W^\star_i, \mbox{ for all $i$ such that } \theta\geq G_i^*(x_i, U^\star_i)\},
\]
where $\min\emptyset =1$ and $\max\emptyset =0$. These functions are clearly non-decreasing and right continuous. 
Note that if $Q_{\bx}(\bU^\star,\bW^\star)\neq\emptyset$, Portmanteau's theorem~\citep{Billingsley1999} and  \eqref{eq:constraint} imply that a distribution function $F\in \overline{Q_{\bx}(\bU^\star,\bW^\star)}$ if and only if  $F^L(\theta)\leq F(\theta)\leq F^U(\theta)$ for all $\theta\in\mathcal S$.
Thus the functions $F^U$ and $F^L$  will be called the upper and lower fiducial bounds throughout.
\yifan{A sample from $F^U$ and $F^L$ can be used to perform estimation and inference for the unknown distribution function $F(\theta)$ in the same way that posterior samples are used in the Bayesian context.}
We generate realizations of $F^U$ and $F^L$ by a novel Gibbs sampler in the next section.

\subsection{Gibbs Sampling and GFD based inference}
We need to generate $(\bU^\star,\bW^\star)$ from the standard uniform distribution 
on a set described by Equation~\eqref{eq:constraint2}, which is achieved by using a Gibbs sampler.
For each fixed $i$, denote random vectors with the $i$-th observation removed by $(\bU^\star_{[-i]},\bW^\star_{[-i]})$. If  $(\bU^\star,\bW^\star)$ satisfy the constraint \eqref{eq:constraint2}, so do $(\bU^\star_{[-i]},\bW^\star_{[-i]})$. The proposed Gibbs sampler is based on the conditional distribution of
\begin{equation}\label{eq:Gibbs}
 (U^\star_i,W^\star_i) \mid \bU^\star_{[-i]},\bW^\star_{[-i]},
 \end{equation}
which is a bivariate uniform distribution on a set $A$, where $A$ is a disjoint \mbox{union} of small rectangles. The beginnings and ends of the rectangles' bases are the neighboring points in the set $\bigcup_{j\neq i}\{G_i^*(x_i,U^\star_j), G_i^*(x_i-1,U^\star_j)\}$, while the corresponding location and height on the vertical axis is determined by \eqref{eq:constraint2}. Details are described in Algorithm~\ref{alg:gibbs} and 
a visualization of the rectangles is shown in Figure~\ref{fig:4}.
Each marginal conditional distribution is supported on the entire $\mathcal S$ and therefore we expect the proposed Gibbs sampler to mix well.

\begin{figure}[H]
\begin{center}
\includegraphics[width=6cm]{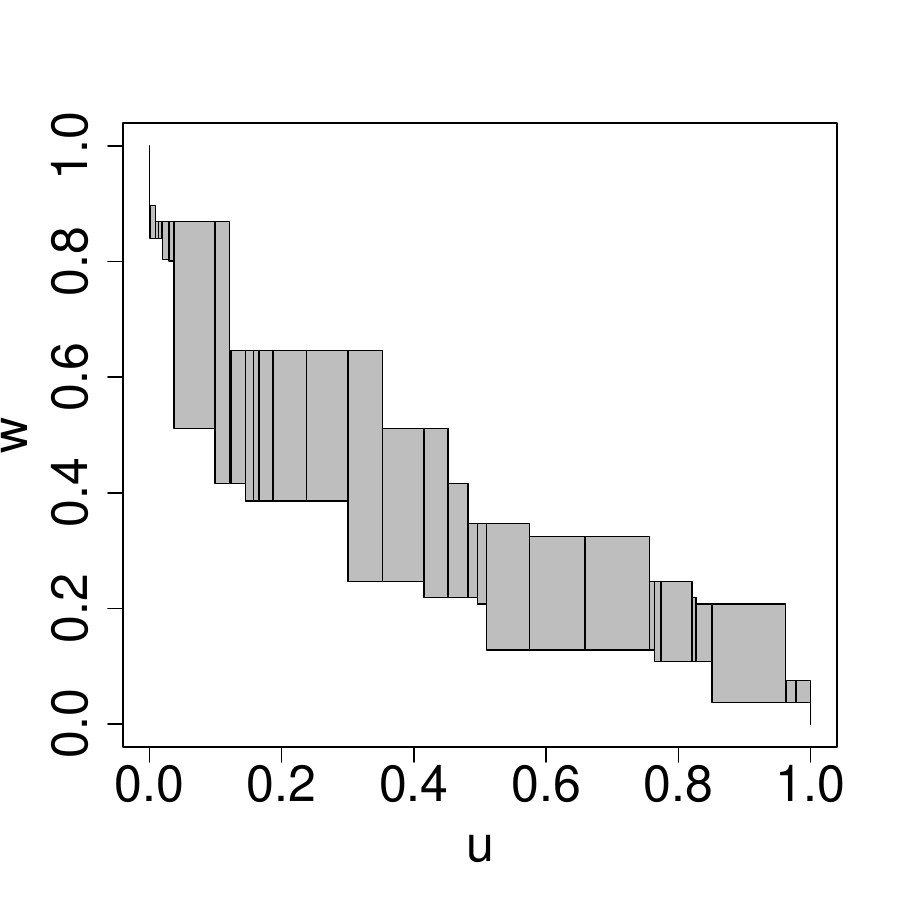}
\includegraphics[width=6cm]{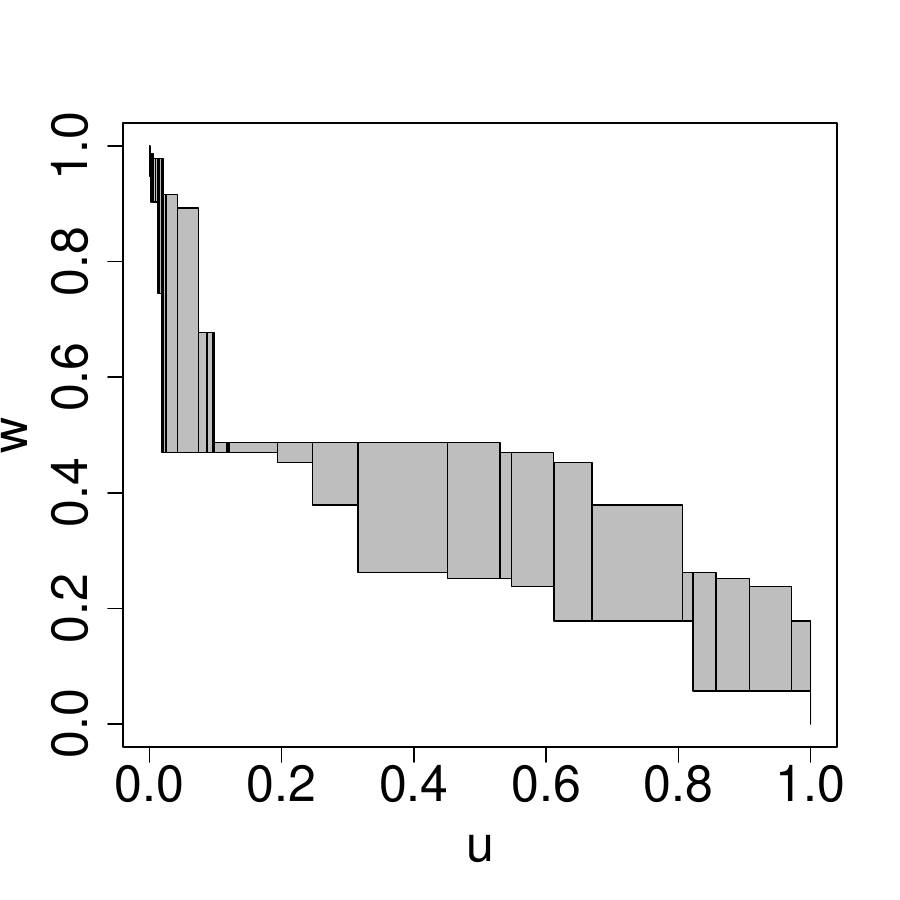}
\end{center}
\caption{\yifan{A visualization of rectangles in Algorithm~1 for the final step of the Gibbs sampler ($j=n_{burn}+n_{mcmc}, i=n$) for the toy example described in Section~\ref{sec:2.3}. The last $(U_n,W_n)$ was generated from the gray area, with the rectangles indicated by black lines. Each panel represents a different starting value, respectively. \label{fig:4}}}
\end{figure}

The proposed Gibbs sampler \yifan{needs} starting points,
and we consider \mbox{two} potential initializations.
The first one starts with randomly generated $(\bU^\star,\bW^\star)$ from independent $\tunif(0,1)$ and reorders $\bW^\star$ so that the constraint~\eqref{eq:constraint2} is satisfied.
The second starting value is consistent with having a deterministic $\Theta$, e.g., $p=\frac{\sum_{i=1}^n x_i}{\sum_{i=1}^n m_i}$ for binomial data $X_i \sim \tbin(m_i,p)$ and $\lambda=\frac{\sum_{i=1}^n x_i}{n}$ for Poisson data $X_i \sim \tpoi(\lambda)$. As these two starting points are very different, they can be used to monitor convergence. To streamline our presentation, in Section~\ref{sec:simulation} and \ref{sec:realdata}, we present numerical results using the first initialization.

\begin{algorithm} 
\SetAlgoLined
\caption{Pseudo algorithm for the fiducial Gibbs sampler} \label{alg:gibbs}
\KwIn{Dataset, e.g., $(m_i,x_i)$, $i=1,\ldots, n$ for binomial data,\\
      ~~~~~~~~~~ starting vectors $\bu, \bw$ of length $n$,\\
      ~~~~~~~~~~ $n_{\text{mcmc}}$ , $n_{\text{burn}}$, and vector $\btheta_{\text{grid}}$ of length $n_{\text{grid}}$.}
\For{$i = 1$ \textbf{to} $n$}{
$\btheta_{L}[i]=G_i^{*}(x_i-1,\bu[i]), \btheta_{U}[i]=G_i^{*}(x_i,\bu[i])$\;
}
\ShowLn Run Gibbs Sampler using the initial values $\bu,\bw,\btheta_L, \btheta_U$\;
\For{$j = 1$ \textbf{to} $n_{\text{burn}}+n_{\text{mcmc}}$}
{
\For{$i = 1$ \textbf{to} $n$}
{
\ShowLn $\bu^0=\bu[-i],\bw^0=\bw[-i],\btheta^0_{L}=\btheta_{L}[-i], \btheta^0_{U}=\btheta_{U}[-i]$\;
\ShowLn $\bu^{\text{pre}}_{L}=G_i(x_i,\btheta^0_{L})$, $\bu^{\text{pre}}_{U}=G_i(x_i-1,\btheta^0_{U})$\;
\ShowLn Sort $\bu^{\text{pre}}=(\bu^{\text{pre}}_{L}, \bu^{\text{pre}}_{U},0,1) $, denoted as $\bu^{\text{sort}}$\;
\ShowLn Sort $(\bw^0,1(n-1),1,0)$ according to the order of $\bu^{\text{pre}}$ as $\bw^\star_U$,\\~~where $1(n-1)$ is a vector with elements 1 of length $n-1$\;
\ShowLn Sort $(0(n-1),\bw^0,1,0)$ according to the order of $\bu^{\text{pre}}$ as $\bw^\star_L$,\\~~where $0(n-1)$ is a vector with elements 0 of length $n-1$\;
\ShowLn $\bw^{\text{pre}}_{U}=\text{cummin}(\bw^\star_U)$, $\bw^{\text{pre}}_{L}=\text{rev}(\text{cummax}(\text{rev}(\bw^\star_L)))$\;
\ShowLn Take the component-wise difference of $\bu^{\text{sort}}$, denoted as $\bu^{\text{diff}}$\;
\For{$k = 1$ \textbf{to} $2n-1$}{
\ShowLn $\bw^{\text{diff}}[k]= \bw^{\text{pre}}_{U}[k]- \bw^{\text{pre}}_{L}[k+1]$\;
}
\ShowLn Sample $i^\star \in \{1,\ldots,2n-1\}$ with probability $\propto \bu^{\text{diff}}\cdot \bw^{\text{diff}}$\;
\ShowLn Sample $a$ and $b$ from independent \tunif(0,1), and set \\~~~$u=\bu^{\text{sort}}[i^\star]+\bu^{\text{diff}}[i^\star]\cdot a,
        w=\bw^{\text{pre}}_U[i^\star]-\bw^{\text{diff}}[i^\star]\cdot b$,\\
~~~$\theta_{L}=G_i^{*}(x_i-1,u), \theta_{U}=G_i^{*}(x_i,u)$,\\
~~~$\bu[i]=u,\bw[i]=w,\btheta_{L}[i]=\theta_{L}, \btheta_{U}[i]=\theta_{U}$\;
}
Generate $n$ i.i.d. \tunif(0,1) and sort them according to the order of $\bw$, denoted by $\bw^*$. Replace $\bw$ by $\bw^*$, i.e., $\bw=\bw^*$\;
}
\ShowLn Evaluate the upper and lower bounds on a grid of values $\btheta_{\text{grid}}$ for each MCMC sample after burn-in, indexed by $l$\;
\For{$j = 1$ \textbf{to} $n_{\text{grid}}$}{
$F_l^{L}(\btheta_{\text{grid}}[j]) = \max( \bw[\btheta_{U}\leq \btheta_{\text{grid}}[j] ])$\;
$F_l^{U}(\btheta_{\text{grid}}[j])=  \min( \bw[\btheta_{L}\geq \btheta_{\text{grid}}[j] ])$\;
}
\Return  The fiducial samples $F_l^{U}$, $F_l^{L}$ evaluated on $\btheta_{\text{grid}}$.
\label{algorithm}
\end{algorithm}

From Algorithm~\ref{alg:gibbs}, we output two distribution functions that are needed for the proposed mixture and conservative confidence intervals.
In the rest of this paper, we denote Monte Carlo realizations of the lower and upper fiducial bounds by $F_l^L$ and $F_l^U$, respectively, where $l=1,\ldots, n_{\text{mcmc}}$, and $n_{\text{mcmc}}$ is the number of fiducial samples.

We propose to use the median of the $2n_{\text{mcmc}}$ samples $\{F_l^L(\theta), F_l^U(\theta), l=1,\ldots, n_{\text{mcmc}}\}$ as a point estimator of distribution function $F(\theta)$.
We construct two types of pointwise confidence intervals,  conservative and mixture, using appropriate quantiles of fiducial samples \citep{Hannig2009}.
In particular, the $95\%$ conservative confidence interval is formed by taking the empirical 0.025 quantile of $\{F_l^{L}(\theta),l=1,\ldots, n_{\text{mcmc}}\}$ as the lower limit and the empirical 0.975 quantile of $\{F_l^{U}(\theta),l=1,\ldots, n_{\text{mcmc}}\}$ as the upper limit.
The lower and upper limits of $95\%$ mixture confidence interval are formed by taking the empirical 0.025 and 0.975 quantiles of $\{F_l^L(\theta), F_l^U(\theta), l=1,\ldots, n_{\text{mcmc}}\}$, respectively.


\subsection{Further illustration with a simulated dataset}\label{sec:2.3}
To streamline our presentation, we take the binomial case as our running example hereinafter, i.e., the observed data are $(m_i,x_i),$ and $X_i \sim \tbin(m_i,P_i)$, $i=1,\ldots, n$, where $P_i\in\mathcal S= [0,1]$ plays the role of $\Theta_i$. We also provide the details of the proposed approach and some examples for the Poisson data in the Supplementary Material.

 We present a toy example to demonstrate the proposed fiducial approach. Suppose $F$ follows the Beta distribution $\tbeta(5,5)$. The number of trials $m_i=20$, $i=1,\ldots,n$. The sample size of the simulated binomial data is $n=20$. The fiducial estimates were based on 10000 iterations after 1000 burn-in times.

Figure~\ref{fig:1} presents the last MCMC sample of the lower fiducial bound $F_l^{L}(p)$ (blue line) and upper fiducial bound $F_l^{U}(p)$ (red line) for the two starting points, respectively.
As the fiducial distribution reflects the uncertainty, we do not expect every single fiducial curve to be close to the true CDF (black line).
Furthermore, Figure~\ref{fig:2} presents the mixture (blue line for lower limit; red line for upper limit) and conservative (cyan line for lower limit; magenta line for upper limit) confidence intervals (CIs) with two starting points, respectively computed from the MCMC sample. In addition, we plot the point estimates of the proposed approach along with Efron's $g$-modeling. The brown curve is the fiducial point estimate $\widehat F(p)$. The dashed curve is the point estimate of $F(p)$ for Efron's $g$-modeling without bias correction. Efron's confidence interval with bias correction looks almost the same as without correction thus we omit in the figures.
As can be seen, the proposed fiducial point estimator and confidence intervals capture the shape of the true CDF pretty well.


\begin{figure}[H]
\begin{center}
\includegraphics[width=6.5cm]{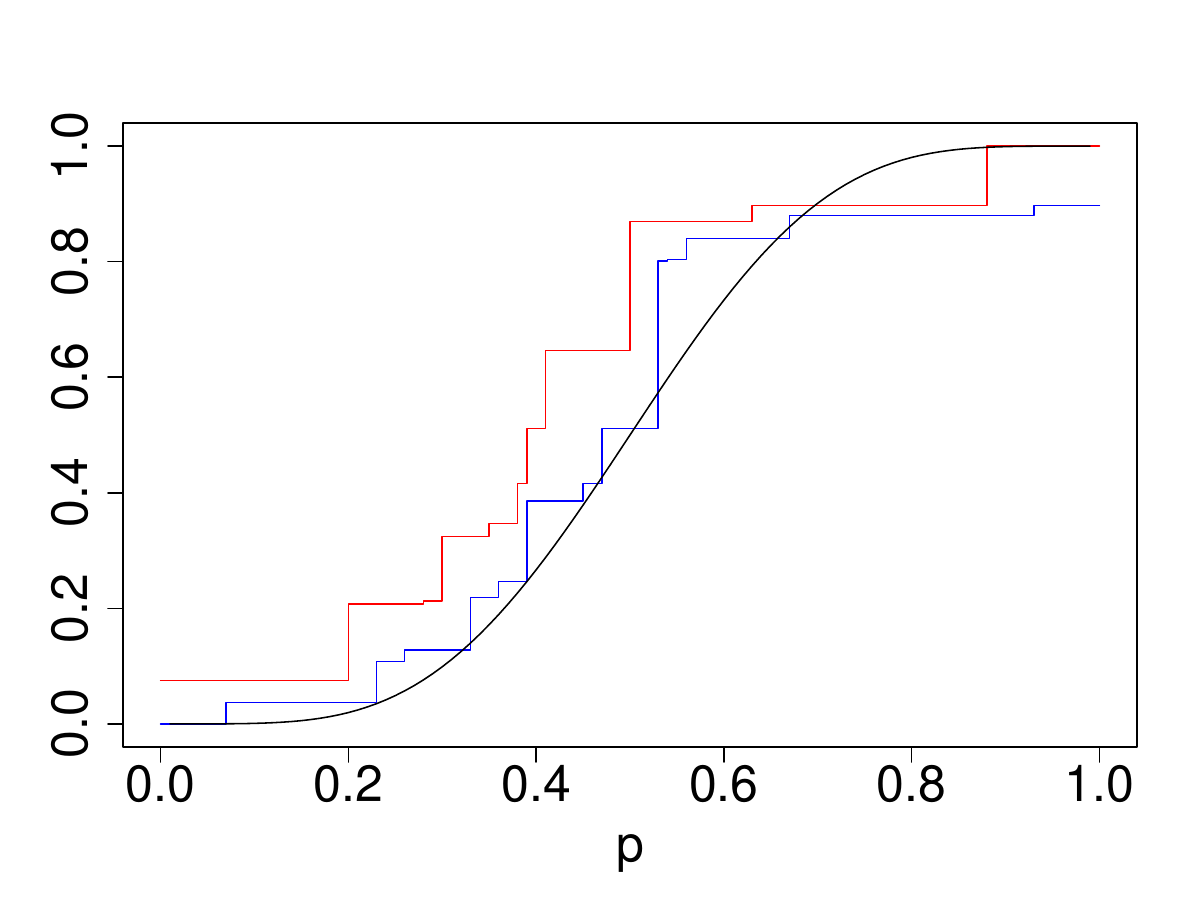}
\includegraphics[width=6.5cm]{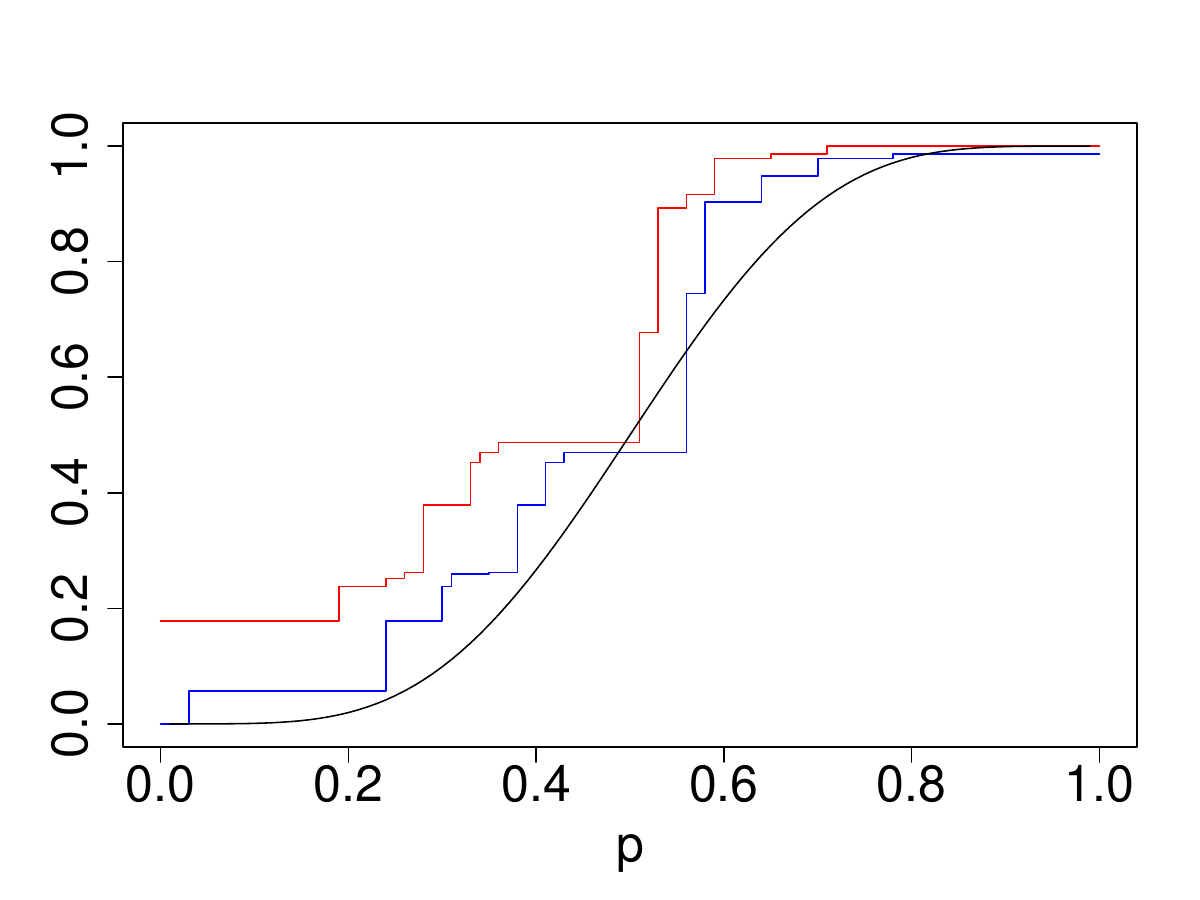}
\label{fig:3}
\end{center}
\caption{The last MCMC sample from \yifan{GFD}. The blue curve is a realization of the lower fiducial bound $F^{L}(p)$ and the red curve is a realization of the upper fiducial bound $F^{U}(p)$. The black curve is the true $F(p)$. Each panel represents a different starting value.}
\label{fig:1}
\end{figure}

\begin{figure}[H]
\begin{center}
\includegraphics[width=6.5cm]{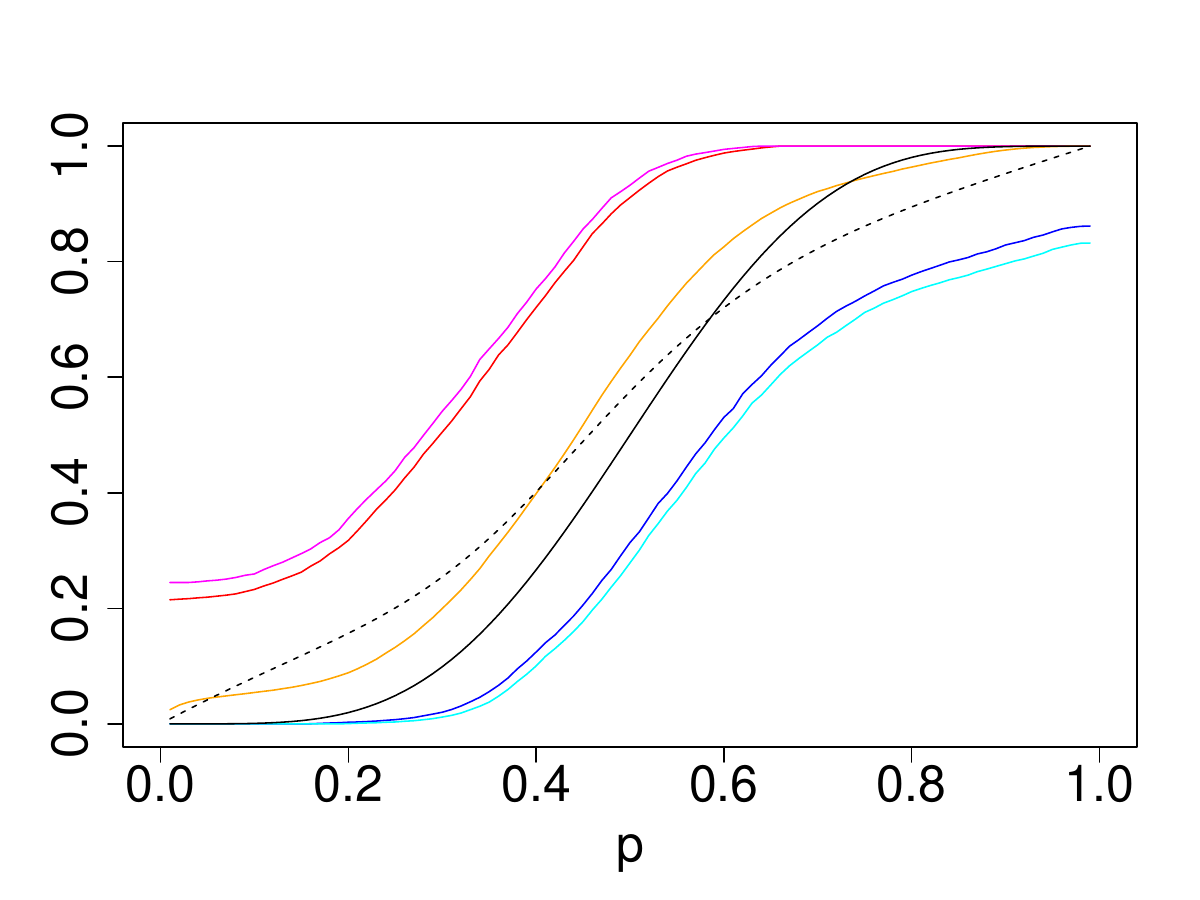}
\includegraphics[width=6.5cm]{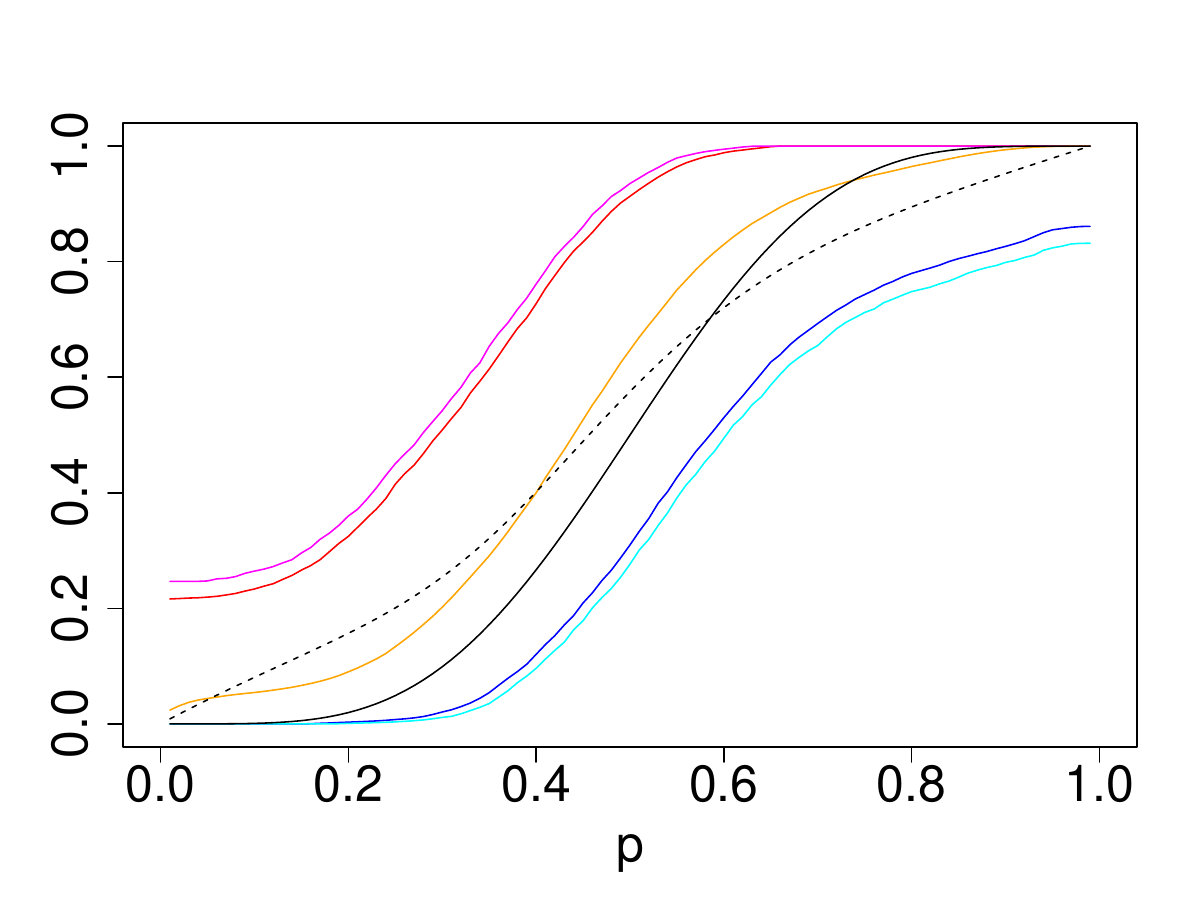}
\end{center}
\caption{Point estimates and 95\% CIs for $F(p)$ given a fixed simulated dataset. Each panel represents an interval computed  from a realization of MCMC chain initiated with different starting values. The orange curve is the fiducial point estimate $\widehat F(p)$.
The dashed curve is the point estimate of $F(p)$ for Efron's $g$-modeling. The black curve is the true $F(p)$. The blue and red curves are lower and upper limits of mixture CIs, respectively.
The cyan and magenta curves are lower and upper limits of conservative CIs, respectively.}
\label{fig:2}
\end{figure}

\section{Theoretical results}\label{sec:theory}
Recall that the GFD is a data-dependent distribution which is defined for every fixed
dataset $\bx$. It can be made into a random measure in the same way as one defines the
usual conditional distribution, i.e., by plugging random variables $\bX$ into the observed dataset.
In this section, we study the asymptotic behavior of this random measure for binomial distribution $X_i \sim \tbin(m_i,P_i)$ when the rate of $m_i$ is much faster than $n$, i.e., the following assumption holds.
\begin{assumption}
$\lim_{n \rightarrow \infty} n^4(\log n)^{1+\epsilon}/ (\min_{i=1,\ldots,n} m_i) = 0$ for any $\epsilon>0$.
\label{as:2}
\end{assumption}

We provide a central limit theorem for $F^L(p)$. The proof of Theorem~\ref{main} is deferred to Section~\ref{thm:3.1} in the Appendix. A similar result holds for $F^U(p)$.
\begin{theorem}\label{main}
Suppose \yifan{the true CDF $F$} is absolutely continuous with a bounded density.
Based on Assumption~\ref{as:2},
\begin{equation}
n^{1/2}\{ F^L(\cdot)-\widehat F_n(\cdot)  \} \rightarrow B_F(\cdot),
\label{eq:gaussian}
\end{equation}
 in distribution on Skorokhod space $\mathcal D[0,1]$ in probability, where 
 \begin{equation}\label{eq:OracleECDF}
\widehat F_n(s) \equiv \frac{1}{n}\sum_{i=1}^n I[P_i\leq s]
\end{equation}
is the oracle empirical CDF constructed based on unobserved $P_i$ that were used to generate the observed $X_i,\ i=1,\ldots,n$, and $B_F(\cdot)$ is a mean zero Gaussian process with covariance $\operatorname{cov}(B_F(s),B_F(t))=F(t\wedge s)-F(t)F(s)$.
\end{theorem}

Notice that the stochastic process on the left-hand side of \eqref{eq:gaussian} is naturally in $\mathcal D[0,1]$, the space of functions on $[0,1]$ that are right continuous and have left limits. Distances on $\mathcal D[0,1]$ are measured using  Skorokhod's metric which makes it into a Polish space \citep{Billingsley1999}. To understand the mode of convergence used here, note that there are two sources of randomness present. One is from the fiducial distribution itself
that is derived from each fixed data set. The other is the usual randomness of the data. Thus \eqref{eq:gaussian} can be interpreted as 
\begin{equation}\label{eq:explanation123}
 \rho\left(n^{1/2}\{ F^L(\cdot)-\widehat F_n(\cdot)\},B_F(\cdot)\right) \overset{pr}{\to} 0,
\end{equation}
where $\rho$ is any metric metrizing weak convergence of probability measures on the Polish space $\mathcal D[0,1]$, e.g., L\'evy-Prokhorov or Dudley metric \citep{Shorack2017}. The distribution of $n^{1/2}\{ F^L(\cdot)-\widehat F_n(\cdot)\}$ in the argument of $\rho$ is the fiducial distribution, i.e., induced by the randomness of $(\bU^\star,\bW^\star)$ with the data $X_i$ and $P_i$ being fixed. Consequently, the left-hand side of \eqref{eq:explanation123} is a function of $X_i$ and $P_i$, and the convergence in probability is based on the distribution of the data.

Theorem~\ref{main} establishes a Bernstein-von Mises theorem for the fiducial distribution under Assumption~\ref{as:2}, which implies that confidence intervals described in Section~\ref{sec:method} have asymptotically correct coverage. Moreover, Theorem~\ref{main} provides a sufficient condition for $n^{1/2}$-estimability of binomial probability parameter's distribution function.
While Assumption~\ref{as:2} is pretty stringent, it does not seem
likely that one can establish a unified asymptotic property of the proposed fiducial approach under a general scheme.
This is best seen by the fact that \yifan{in the binomial case if} the $m_i$ are uniformly bounded, there is not enough information in the data to consistently estimate the underlying distribution function of $P$.

Interestingly, looking at the fiducial solution reveals an interesting connection to a different statistical problem \citep{cui2021}. Recall that the quantities $P_i,\ i=1,\ldots,n$ are only known to be inside random intervals $(G_i^{*}(x_i-1,U_i^\star), G_i^{*}(x_i,U_i^\star)]$. Therefore, the statistical problem we study here bears similarities to non-parametric estimation under Turnbull's general censoring scheme \citep{efron1967two,turnbull1976empirical}, in which case there is no unified theory of the nonparametric maximum likelihood estimator but properties are investigated under various special and challenging cases such as $n^{1/2}$ convergence for right-censored data \citep{breslow1974large}, and $n^{1/3}$ convergence for current status data \citep{groeneboom1992information}.

\begin{remark}\label{newremark}
Suppose $m_i\equiv m$. Notice that $$pr(F \in Q_x(U^*,W^*))\propto  \prod_{i=1}^{n} \left[\int_0^1  {{m}\choose{x_i}} p^{x_i} (1-p)^{m-x_i} dF(p) \right] $$ is proportional to the nonparametric likelihood function. 
\yifan{A simple calculation shows that maximizing the scaled fiducial probability in its limit provides exactly the underlying true CDF. Detailed derivations and further discussions of this observation are provided in Section
~\ref{sec:moments} of the Supplementary Material. 
}
\end{remark}

\section{Numerical experiments}\label{sec:simulation}

We perform simulation studies to compare the frequentist properties of the proposed fiducial confidence intervals with Efron's $g$-modeling \citep{10.1093/biomet/asv068,deconvolveR},  the nonparametric bootstrap, and a nonparametric Bayesian approach \citep{dirichletprocess}.
For each scenario, we first generated $p_i, i=1,\ldots, n$ from the distribution function $F$. Then we drew $X_i$ from the binomial distribution Bin($m_i,p_i$), where $m_i$ is described in Section~\ref{subsec:simulation}. The simulations were replicated 500 times for each scenario.

For the proposed method as well as other existing methods, we choose the grid $p=0.01, 0.02, \ldots, 0.99$ following \cite{narasimhan2016g}.
The fiducial estimates were based on 2000 iterations of the Gibbs sampler after 500 burn-in times.
Efron's $g$-modeling was implemented using R package \texttt{deconvolveR} \citep{deconvolveR}. We used default values for the degree of the splines, i.e., 5.
We considered the regularization strategy with
the default value $c_0=1$.
For the nonparametric bootstrap, we first obtained the maximum likelihood estimates $\hat p_i=x_i/m_i$. We then constructed the empirical CDF as the point estimator, and used $B=1000$ bootstrap samples of $\hat p_i$ to construct confidence intervals. We also considered a fully Bayesian approach, i.e., Dirichlet process mixture of Beta Binomial, which gives more flexibility than a Beta binomial model \cite[p19]{ross2018dirichletprocess}. Default values for the prior parameters in R package \texttt{dirichletprocess} \citep{dirichletprocess} were used. The Bayesian estimates were based on 2000 MCMC samples after 500 burn-in times.

\subsection{Simulation settings}\label{subsec:simulation}

We start with the following two scenarios with \yifan{$m_i$ being the same across individuals}.
\begin{scenario} \normalfont
We consider the same setting as Section~\ref{sec:2.3}.
Let \yifan{$\Theta$} follow Beta distribution $\tbeta(5,5)$, and the number of trials $m_i=20$, $i=1,\ldots,n$. The sample size $n$ of the simulated binomial data was set to 50.
\end{scenario}

\begin{scenario} \normalfont
Let \yifan{$\Theta$} follow a mixture of Beta distributions $0.5\tbeta(10,30)+0.5\tbeta (30,10)$, and the number of trials $m_i=20$, $i=1,\ldots,n$. The sample size $n$ of the simulated binomial data was set to 50.
\end{scenario}

Next, we consider three complex settings from \cite{10.2307/41724531}.

\begin{scenario} \normalfont(Beta density)
We let \yifan{$\Theta$} follow $\tbeta(8,8)$, and the $m_i$'s are integers sampled uniformly between 100 and 200. The sample size $n$ of the simulated binomial data equals to 100.
\end{scenario}

\begin{scenario} \normalfont (A multimodal distribution)
We consider a mixture of Beta distributions $0.5\tbeta(60,10)+0.5\tbeta(10,60)$ for \yifan{$\Theta$}. The sample size $n$ of the simulated binomial data equals to 100, and the $m_i$'s equal to 100, $i=1,\ldots,n$.
\end{scenario}

\begin{scenario} \normalfont (Truncated exponential)
Let \yifan{$\Theta$} follow $\texp(8)$ truncated at 1. The simulated binomial data are of size $n = 200$ and $m_i = 100$ for $i = 1, \ldots,n.$
\end{scenario}

We also consider the above five settings for $n= 1000$. The results are reported in the Supplementary Material.

\subsection{Numerical results}

In this section, we first compare the mean squared error (MSE) of different methods of $F(p)$ for the five scenarios.
The numerical results for $p = 0.15, 0.25, 0.5, 0.75, 0.85$ for each scenario are presented in Table~\ref{table:mse}.
We see that the \mbox{MSEs} of the proposed fiducial point estimates are as good as and sometimes smaller than competing methods.


\begin{table}[!h]
\centering
\scalebox{1.1}{
\begin{tabular}{ccccccc}
  \hline
 Scenario & $p$ & F & g & bc & BP & BA \\
  \hline
      \multirow{2}{*}{1} & 0.15 & 5 &  40  & 39  & 17  &  1\\
& 0.25 &  20 &  49 &  48 &  69  & 10\\
& 0.50 & 51 &  47  & 48 &  71 &  78\\
&0.75 & 18  & 22  & 21  & 16  & 10\\
& 0.85 & 5  & 13 &  12   & 4   & 1\\
   \hline
    \multirow{2}{*}{2} & 0.15 &  41 & 149 & 149 & 145 &   9  \\
&0.25 & 39 &  17 &  18  & 66 & 101\\
&0.50 & 49 &  32  & 33  & 54 &  53\\
&0.75 & 37 &  18 &  19  & 50 &  88\\
&0.95 & 46 &  86 &  84  & 33  & 12\\
   \hline
   \multirow{2}{*}{3}  & 0.15 &   0.14  & 8.45 & 8.16 &  0.15 &  0.04\\
& 0.25  & 2.57 & 14.89 & 14.41 & 2.87 & 1.39\\
& 0.50 & 22.85 &  26.94 & 27.04 &  24.68 &  21.92\\
&0.75  &  2.54 & 5.03 & 4.81 & 2.80 & 1.34\\
&0.85   & 0.16  & 1.46 & 1.38  & 0.15 &   0.05\\
   \hline
 \multirow{2}{*}{4}  & 0.15 &  22 &  49  & 49 &  23 &  27\\
& 0.25 & 27 &  34 &  33 &  26 &  25\\
&0.50 & 25  & 24  & 25  & 26  & 26\\
& 0.75 & 26 &  36  & 35  & 27  & 25\\
& 0.85 & 20 &  52 &  52 &  26  & 25\\
   \hline
  \multirow{2}{*}{5}  &   0.15  & 11  & 10  & 10 &  11 &  10* \\
     & 0.25 & 6 &  10 &  10  &  6   & 6* \\
     & 0.50 & 1   & 5  &  4  &  1   & 1* \\
     & 0.75 &  0.13 & 1.84 & 1.66 &  0.12 &  0.10* \\
     & 0.85 & 0.05  & 0.76 & 0.68 &  0.05  &  0.04* \\
   \hline
\end{tabular}
}
\caption{MSE ($\times 10^{-4}$) of point estimates for $F(p)$ of each scenario.
``F'' denotes the fiducial point estimates;
``g'' denotes Efron's $g$-modeling without bias correction;  ``bc'' denotes Efron's $g$-modeling with bias correction;
``BP'' denotes the bootstrap method; ``BA'' denotes the Bayesian method.
*The Bayesian results for Scenario 5 are reported based on 489 replications as 11 runs failed due to an error in R package \texttt{dirichletprocess}.}
\label{table:mse}
\end{table}

\begin{table}[!h]
\centering
\scalebox{1.1}{
\begin{tabular}{cccccccc}
  \hline
 Scenario & $p$ & M & C & g & bc & BP & BA \\
  \hline
      \multirow{2}{*}{1} &0.15 & 99 & 100  & 25  & 27 &  74 &  88\\
&0.25 & 99 & 100  & 80 &  82 &  65  & 83\\
&0.50 & 100 & 100  & 94 &  94 &  91  & 84\\
&0.75 & 100  & 100  & 96  & 96  & 92  & 83\\
&0.85 & 100 & 100 &  90 &  91  & 54  & 89\\
   \hline
    \multirow{2}{*}{2}  & 0.15 &  98  & 99  &  1  &  1 &  27  & 97\\
&0.25 & 100 & 100  & 98 &  98  & 90   & 86\\
&0.50& 98  & 98  & 96 &  96 &  97  & 45\\
&0.75 &  100 & 100 &  96 &  96  & 85  & 87\\
& 0.85 &  97  & 99  & 31 &  31 &  79  & 97\\
   \hline
 \multirow{2}{*}{3}   & 0.15&   100  & 100  &  4  &  5  & 13 &  28\\
& 0.25 &    99 &  99  & 69  & 71  & 86  & 28\\
& 0.50 & 99 &  99  & 89 &  89 &  96  & 27\\
&0.75 &  98  & 99 &  97 &  97  & 87 &  32\\
& 0.85 &  99 & 100  & 97  & 98  & 12  & 31\\
   \hline
 \multirow{2}{*}{4}   & 0.15&   99 & 100  & 48  & 49  & 95  & 59\\
& 0.25 &  95  & 97 &  89 &  90  & 93 &  12 \\
& 0.50 & 95 &  96  & 93  & 93&   95  &  6\\
&0.75 & 95 &  95  & 89  & 89 &  93  &  9\\
& 0.85 &   99 &  99  & 75  & 75   &91  & 59\\
   \hline
  \multirow{2}{*}{5}  &   0.15 &  99  & 99  & 95  & 95  & 93 &  18* \\
     & 0.25 & 97  & 98 &  95  & 95  & 93  & 21* \\
     & 0.50 & 98 &  98 &  98  & 98 &  89  & 20* \\
     & 0.75 & 98  & 99 & 100 & 100  & 36  & 20* \\
     & 0.85 & 99 & 100 & 100 & 100 &  17 &  11* \\
   \hline
\end{tabular}
}
\caption{Coverage (in percent) of 95\% CIs for $F(p)$ of each scenario. ``M'' denotes mixture GFD confidence intervals;  ``C'' denotes conservative GFD confidence intervals;
``g'' denotes Efron's $g$-modeling without bias correction;  ``bc'' denotes Efron's $g$-modeling with bias correction;
``BP'' denotes the bootstrap method; ``BA'' denotes the Bayesian method.
*The Bayesian results for Scenario 5 are reported based on 489 replications as 11 runs failed due to an error in R package \texttt{dirichletprocess}.}
\label{table:c}
\end{table}

\begin{table}[!h]
\centering
\scalebox{1.1}{
\begin{tabular}{cccccccc}
  \hline
 Scenario & $p$ & M & C & g & bc & BP & BA \\
  \hline
      \multirow{2}{*}{1} & 0.15 &  123 & 139 & 101 & 101 &  84 &  28\\
&  0.25 &  220 & 241 & 171 & 171 & 172 &  95\\
& 0.50 & 426 & 465 & 246  &246 & 272  & 279\\
&0.75 & 217 & 239 & 154  &154  &128 &  89\\
& 0.85 & 121 & 137 &  82 &  81 &  42 &  28\\
   \hline
    \multirow{2}{*}{2} &0.15 & 243 & 266  & 66&   66 & 185 & 101 \\
& 0.25 & 357 & 390 & 162 & 162 & 251 & 311\\
&0.50 & 305  & 330 & 206  &206  &275  & 90\\
& 0.75 & 352 & 385 & 166 & 166& 226 & 301\\
& 0.85 &245 & 268 & 143 & 143 & 133  & 107\\
   \hline
   \multirow{2}{*}{3}  & 0.15 &  35  & 42  & 46  & 46  &  4  &  2\\
& 0.25 &  77 &  84  & 81  & 81  & 51  & 10\\
& 0.50 & 243  & 259 & 168 & 168  & 195 &  33\\
& 0.75  & 78 &  86 &  65  & 65  & 51  & 10\\
 & 0.85 & 35 &  42  & 27  & 26  &  4   & 2\\
   \hline
 \multirow{2}{*}{4}  &0.15&  240 & 259 & 125 & 125 & 179 & 101\\
& 0.25 & 201 & 213 & 191 & 191  &195 &  16\\
&0.50&  193 & 202 & 188  &188 & 196  &  0.02\\
&0.75& 201 & 213 & 198 & 198  &195  & 16\\
&0.85& 240 & 259 & 178 & 178 & 173 &  98\\
   \hline
  \multirow{2}{*}{5}  &   0.15 & 160 & 171 & 125 & 125 & 126 &  18* \\
     & 0.25 & 115 & 124 & 114 & 114  & 94 &  13* \\
     & 0.50 & 45 &  50  & 74 &  73 &  35  &  6*\\
     & 0.75 &  20 &  24  & 37 &  36 &   7 &  1*\\
     & 0.85 & 17 &  20  & 23  & 22  &  3 &   1* \\
   \hline
\end{tabular}
}
\caption{Mean length ($\times 10^{-3}$) of 95\% CIs for $F(p)$ of each scenario.  ``M'' denotes mixture GFD confidence intervals;  ``C'' denotes conservative GFD confidence intervals;
``g'' denotes Efron's $g$-modeling without bias correction;  ``bc'' denotes Efron's $g$-modeling with bias correction;
``BP'' denotes the bootstrap method; ``BA'' denotes the Bayesian method.
*The Bayesian results for Scenario 5 are reported based on 489 replications as 11 runs failed due to an error in R package \texttt{dirichletprocess}.}
\label{table:l}
\end{table}

Next, we present the coverage and average length of confidence intervals for various methods in Tables~\ref{table:c}-\ref{table:l}.
Table~\ref{table:c} summarizes the coverage of 95\% confidence intervals of various methods, and Table~\ref{table:l} summarizes the average length of these confidence intervals. ``M'' denotes mixture GFD confidence intervals;  ``C'' denotes conservative GFD confidence intervals;
``g'' denotes Efron's $g$-modeling without bias correction;  ``bc'' denotes Efron's $g$-modeling with bias correction;
``BP'' denotes the bootstrap method; ``BA'' denotes the Bayesian method.

\yifan{
For point estimators, overall, the proposed fiducial method and the Bayesian method outperform other methods. The performance of the proposed estimator and the Bayesian estimator is comparable, with the former outperforming the latter for medium values of $p$ in Scenarios~1-2 and the latter outperforming the former for small and large values of $p$ in Scenarios~1-3.
For uncertainty quantification,} we see that the GFD confidence intervals maintain or exceed the nominal coverage everywhere,
while other methods often have coverage problems.
In particular, Efron's confidence intervals and nonparametric bootstrap have substantial coverage problems close to the boundary, while the Bayesian method consistently underestimates the uncertainty resulting in credible intervals that are too narrow.

\yifan{
It is not surprising that the GFD confidence intervals are often longer than other methods as the GFD approach aims to provide a conservative way to quantify uncertainty.
As expected, the mean length of mixture GFD confidence intervals is a little shorter than conservative GFD confidence intervals.
A potential reason for the fiducial approach outperforming Efron's $g$-modeling in terms of coverage is that Efron's $g$-modeling relies on an exponential family parametric model and the proposed fiducial approach is nonparametric. 
Therefore, the proposed method is demonstrably robust to certain model mis-specifications, e.g., when the true model does not belong to an exponential family.
}

\section{Intestinal surgery data}\label{sec:realdata}
In this section, we consider an intestinal surgery study on gastric adenocarcinoma involving $n = 844$ cancer patients \citep{gholami2015number}.
Resection of the primary tumor with appropriate dissection of surrounding lymph \mbox{nodes} is the foundation of curative care.
In addition to the primary tumor, surgeons also remove satellite nodes for later testing.
Efron's deconvolution was used to estimate the prior distribution of the probability of one satellite being malignant in this study \citep{gholami2015number}.

The dataset consists of pairs $(m_i,X_i), i=1,\ldots,n$, where $m_i$ is the number of satellites removed and $X_i$ is the number of these satellites found to be
malignant. The $m_i$ varies from 1 to 69. Among all cases, 322 have $X_i = 0$. For the rest of them, $X_i /m_i$ has an approximate $\tunif(0,1)$ \yifan{distribution} \citep{10.1093/biomet/asv068}. We are interested in estimating distribution function of the probability of one satellite being malignant.
Following the model proposed in \cite{10.1093/biomet/asv068}, we assume a binomial model, i.e., $X_i \sim \tbin(m_i,P_i)$, where $P_i$ is the $i$-th patient's probability of any one satellite being malignant.

We compared the proposed mixture and conservative GFD confidence intervals to Efron's with and without bias correction, the bootstrap method, and the fully Bayesian approach. For all methods, we used the grid $[0.01, 0.02,$ $\ldots, 0.99]$ for the discretization of $p$. The fiducial and Bayesian estimates were based on 10000 iterations after 1000 burn-in times.
Other tuning parameters for each method were chosen in the same way as Section~\ref{sec:simulation}.

\begin{figure}[!h]
\begin{center}
\includegraphics[width=11cm]{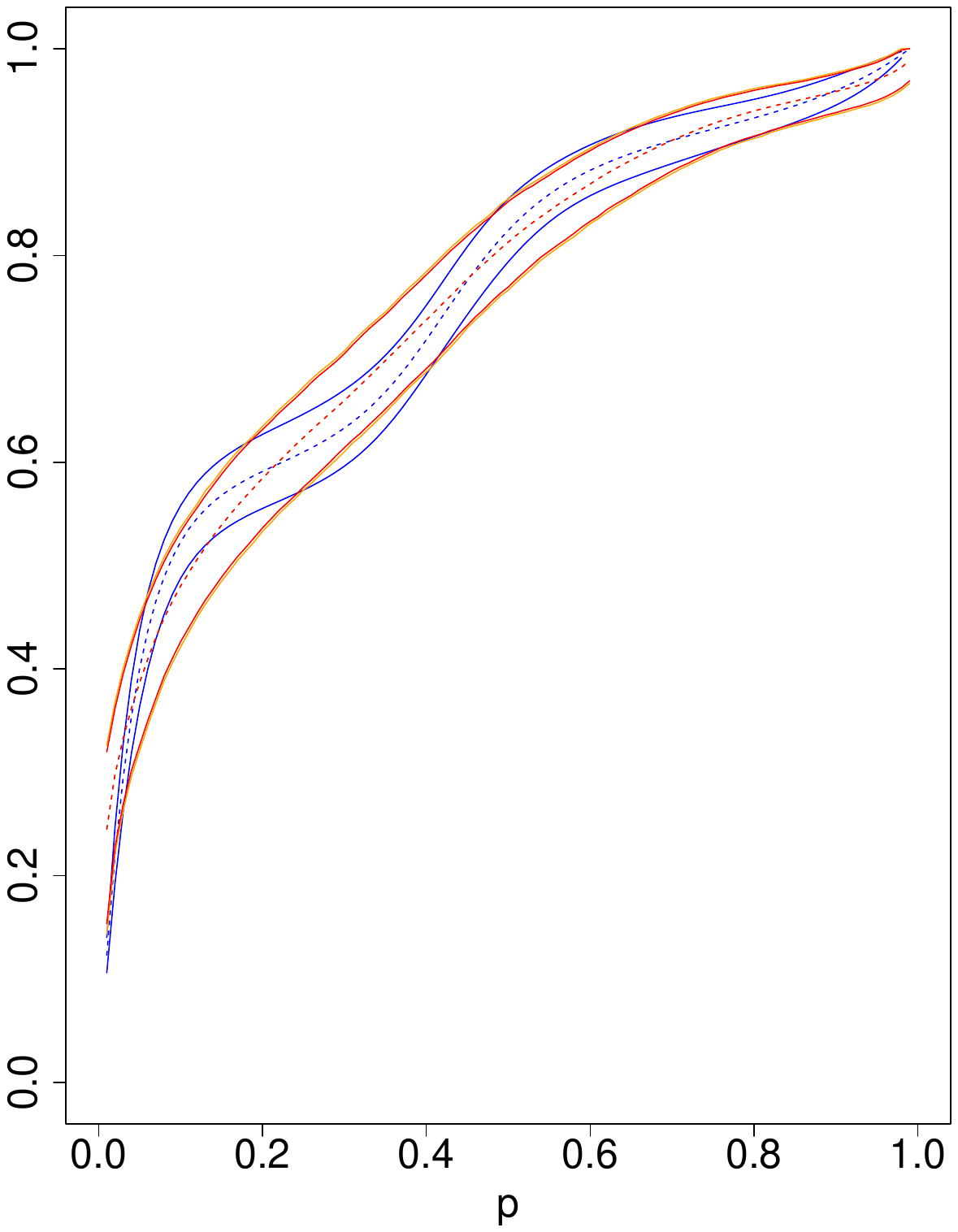}
\end{center}
\caption{GFD versus $g$-modeling.
Estimated CDF (dashed line) and 95\% CIs for $F(p)$  of GFD and Efron's $g$-modeling. The red and orange curves are mixture and conservative confidence intervals, respectively; The blue curve is Efron's confidence intervals without bias correction. Efron's confidence intervals with bias correction looks almost the same as without correction thus we omit in the figure.
}
\label{pic:real}
\end{figure}

The overall shapes of bootstrap and Bayesian confidence intervals are similar to the fiducial ones but much narrower.
We provide bootstrap and Bayesian point estimates and 95\% confidence intervals in the Supplementary Material.
Figure~\ref{pic:real} shows the point estimates and 95\% confidence intervals of the distribution function $F$ for the proposed GFD approach and Efron's $g$-modeling.
Overall, the GFD confidence interval is more conservative. The GFD confidence intervals cover Efron's almost everywhere.

For the proposed fiducial approach, there is a large mode for the upper fiducial confidence interval near  $p= 0$, which coincides with the fact that about $38\%$ of the $X_i$'s are 0 in the surgery data. However, the Bayesian method and Efron's $g$-modeling seem to quantify uncertainty of this proportion to be lower. One exception is the nonparametric bootstrap which gives the estimation of point mass at zero $0.38$ with 95\% confidence intervals $(0.35,0.41)$. We note that this might be overestimated as $X_i=0$ may correspond to a non-zero probability $p$ especially when $m_i$ is small.

Moreover, the generalized fiducial confidence intervals provide us a unimodal density, while Efron's gives a bimodal density.
We believe that the fiducial as well as bootstrap and nonparametric Bayesian answers are more in line with Efron's observation that for those $X_i \neq 0$ in the surgery data, $X_i /m_i$ has an approximate $\tunif(0,1)$ \yifan{distribution} \citep{10.1093/biomet/asv068}.

\section{Discussion}\label{sec:discussion}
In this paper, we proposed a prior-free approach to nonparametric deconvolution problem, and obtained valid point estimates and confidence intervals. 
This was accomplished through a novel algorithm to sample from the GFD.
The median of the GFD is used as the point estimate, and appropriate quantiles of the GFD evaluated at a given $p$
provide pointwise confidence intervals.
We also studied the theoretical properties of the fiducial distribution. Extensive simulations show that the proposed fiducial approach is a good alternative to existing methods such as Efron's $g$-modeling.
We applied the proposed fiducial approach to intestinal surgery data to estimate the probability of each satellite site being malignant for patients.

We conclude by listing some open research problems:
\begin{enumerate}
\item
The proposed fiducial method seems to be a powerful nonparametric approach.
It would be interesting to implement it inside other statistical procedures such as tree or random forest models to include covariates \citep{wu2019uncertainty}.

\item This paper focuses on discrete data. The proposed approach can be extended to continuous data, such as $\tnorm(\Theta,1)$, where $\Theta$ follows a distribution function $F$.
    This part is currently under investigation.

\item As we can see in simulations, the GFD approach is sometimes over-conservative. It could be possible to consider a
different choice of fiducial samples, such as log-interpolation \citep{cuihannig2019} or monotonic spline interpolation \citep{Taraldsen2019discuss,cuihannig2019rejoinder}.

\item The asymptotic distribution results in Section~\ref{sec:theory} only hold under Assumption~\ref{as:2}.
To investigate a non-$n^{1/2}$ rate of convergence should make a fruitful avenue of future research.

\item It should be possible to use the GFD in conjunction with various functional norms to construct simultaneous confidence bands \citep{cuihannig2019,nair1984confidence, martin2019discuss}.
\end{enumerate}

\section{Acknowledgments}
 The authors are thankful to Prof. Hari Iyer for helpful discussions.
The authors are also thankful to the referees, associate editor, and editor for helpful comments which led to an improved manuscript.
This work was supported in part by the National Natural Science Foundation of China, the Singapore Ministry of Education,  the National Institute of Health, and the National Science Foundation.

\newpage

\begin{center}
{\LARGE \textbf{Appendix}}
\end{center}

\appendix
\section{Lemmas \label{thm:0}}
Recall that the observed data points $x_i$ are integers.
\begin{lemma}
$x_i=G_i^{-1}(u_i,\theta_i)$ if and only if $G_i^{*}(x_i-1,u_i) < \theta_i \leq G_i^{*}(x_i,u_i)$.
\label{lemma:equiv}
\end{lemma}

\begin{proof}

Recall the definition of $G_i^{-1}(u_i,\theta_i)= \inf\{x_i: G_i(x_i,\theta_i)\geq u_i\}$. So $x_i=G_i^{-1}(u_i,\theta_i)$ if and only if $G_i(x_i,\theta_i)\geq u_i >G_i(x_i-\epsilon,\theta_i)$ for all $\epsilon>0$. Then by the definition of $G_i^{*}$, it is further equivalent to $G_i^{*}(x_i-1,u_i) < \theta_i \leq G_i^{*}(x_i,u_i)$.
\end{proof}

\begin{lemma}
$Q_{\bx}(\bu,\bw)\neq\emptyset$ if and only if $\bu, \bw$ satisfy: whenever $G_i^*(x_i,u_i)\leq G_j^*(x_j-1,u_j)$ then $w_i< w_j$.
\label{lemma:equiv2}
\end{lemma}

\begin{proof}
Sufficiency: If $Q_{\bx}(\bu,\bw)\neq\emptyset$ holds, and $G_i^*(x_i,u_i)\leq G_j^*(x_j-1,u_j)$, then we know that $ w_i \leq F(G_i^*(x_i,u_i))\leq F(G_j^*(x_j-1,u_j))< w_j $.

Necessity: We prove it by contradiction. If $Q_{\bx}(\bu,\bw)$ is empty, then there must exist indices $i$ and $j$ such that, $(G_j^*(x_j-1,u_j),G_j^*(x_j,u_j)]$ is strictly larger than $(G_i^*(x_i-1,u_i),G_i^*(x_i,u_i)]$ but $w_i \geq w_j$. This contradicts with whenever $G_i^*(x_i,u_i)\leq G_j^*(x_j-1,u_j)$ then $w_i< w_j$.
\end{proof}

\section{Proof of Theorem~\ref{main} \label{thm:3.1}}
\begin{proof}
Recall that the data generating equation is
\begin{equation}\label{eq:DGEinApp}
X_i= G_i^{-1}(U_i,P_i),
\quad 
P_i=F^{-1}(W_i),
\end{equation}
where $G_i$ is the CDF of binomial distribution.
Define the oracle fiducial distribution based on unobserved $P_1,\ldots,P_n$ as
\begin{equation}\label{eq:OracleFidL}
\widetilde F(s)=\sum_{i=0}^n I[
P_{(i)} \leq s< P_{(i+1)}]  W^*_{(i)},
\end{equation}
where $P_{(0)} \equiv 0$, $P_{(n+1)} \equiv 1$, $W^*_{(i)}$ are uniform order statistics, and $W^*_{(0)} \equiv 0$.
Notice \eqref{eq:OracleFidL} is the lower fiducial distribution $F^L$ in \cite{cuihannig2019} when there is no censoring.
By Corollary
~\ref{corollary} of the Supplementary Material, we have  
\begin{align*}
n^{1/2}\{ \widetilde F(\cdot)-\widehat F_n(\cdot)  \} \rightarrow \{1-F(\cdot)\} B(\gamma(\cdot)),
\end{align*}
in distribution on Skorokhod space $\mathcal D[0,1]$ in probability, where $B$ is the Brownian Motion, $\gamma(t)=\int_0^t \frac{f(s)}{[1-F(s)]^2} ds= \frac{F(t)}{1-F(t)}$, and $\widehat F_n$ is the oracle empirical distribution function defined in \eqref{eq:OracleECDF}. 
We also define
\begin{align*}
\widetilde F^L(s)\equiv \sum_{i=0}^n I[G^U_{(i)}\leq s < G^U_{(i+1)}] W^*_{(i)},
\end{align*} 
and
\begin{equation*}
\widetilde F^U(s)\equiv \sum_{i=0}^n I[G^L_{(i)}\leq s < G^L_{(i+1)}] W^*_{(i+1)},
\end{equation*}
where $G^U_{(1)},\ldots,G^U_{(n)}$ are order statistics of $\{G_i^*(x_i,U_i^\star),i=1,\ldots,n\}$,  $G^L_{(1)},\ldots,G^L_{(n)}$ are order statistics of $\{G_i^*(x_i-1,U_i^\star),i=1,\ldots,n\}$, $G^U_{(0)} \equiv 0$, $G^U_{(n+1)} \equiv 1$, $G^L_{(0)} \equiv 0$, $G^L_{(n+1)} \equiv 1$, and  $W^*_{(n+1)} \equiv 1$. Note that $\widetilde F^L$ and $\widetilde F^U$ can be regarded as lower and upper bounds of $F^L$ and $F^U$, respectively.

In order to obtain
\begin{align*}
n^{1/2} \{F^L(\cdot) -\widehat F_n(\cdot)\} \rightarrow \{1-F(\cdot)\} B(\gamma(\cdot)), 
\end{align*}
in distribution in probability, it is enough to show that 
\begin{align*}
\sup_s n^{1/2} |\widetilde F(s)- F^L(s) | \rightarrow 0,
\end{align*}
in probability, which would be implied by 
\begin{align}\label{eq:conv}
\sup_s n^{1/2} \{\widetilde F^U(s)- \widetilde F^L(s) \} \rightarrow 0,
\end{align}
in probability. In order to show Equation~\eqref{eq:conv}, one essentially needs to show for any $\epsilon>0$,
\begin{align*}
pr(\sup_s n^{1/2} \{\widetilde F^U(s)- \widetilde F^L(s) \}>\epsilon ) \rightarrow 0.
\end{align*}
By 
Lemma~\ref{lemma:nointersect}, 
given in the Supplementary Material, there is no intersection between $\{(G_i^*(x_i-1,U_i^\star),G_i^*(x_i,U_i^\star)],i=1,\ldots,n\}$ with a large probability converging to 1.

Thus, we have that
\begin{align*}
 pr(\sup_s n^{1/2} \{\widetilde F^U(s)- \widetilde F^L(s) \}>\epsilon ) 
\leq & \sum_{i=0}^n pr(W^*_{(i+1)}-W^*_{(i)} > \frac{\epsilon}{n^{1/2}})  \\ 
= & (n+1) \times pr(Beta(1,n)>\frac{\epsilon}{n^{1/2}} )\\
= & (n+1) \times (1-\frac{\epsilon}{n^{1/2}})^n
\rightarrow  0.
\end{align*}

Therefore, we have 
\begin{align*}
n^{1/2} \{F^L(\cdot) -\widehat F_n(\cdot)\} \rightarrow \{1-F(\cdot)\} B(\gamma(\cdot)),
\end{align*}
 in distribution in probability. Note that for any $t<s$,
$$\operatorname{cov}[\{1-F(s)\}B(\gamma(s)),\{1-F(t)\} B(\gamma(t))]=\gamma(t)\{1-F(s)\}\{1-F(t)\}=F(t)\{1-F(s)\},$$ which completes the proof.
\end{proof}

\newpage

\begin{center}
{\LARGE \textbf{Supplementary Material}}
\end{center}

\section{Lemma~\ref{lemma:nointersect0}  and its proof}
\begin{lemma}\label{lemma:nointersect0}
Assume the conditions of Theorem~\ref{main}. Suppose that $X_i \sim
\tbin(m_i
, P_i)$, where $P_i$ are unobserved i.i.d. random variables with distribution function $F$.
We have that \begin{align}\label{eq:xm}
\min_{i,j\in\{0,\ldots,n\}} \left \{\frac{X_i}{m_i}-\frac{X_j}{m_j}\right\} \succeq O\left(\frac{1}{n^2(\log n)^{\sqrt{\epsilon/2}}}\right),
\end{align}
for any $\epsilon>0$, with a large probability converging to 1.
\end{lemma}

\begin{proof}
We first show that the unobserved $P_i=F^{-1}(W_i)$ are well separated.
Straightforward calculation with uniform order statistics shows that
\[pr\left(\min_{i\in\{0,\ldots,n\}} \left \{W_{(i+1)}-W_{(i)}\right\} >\frac{t}{n(n+1)}\right)\geq \left(1-\frac{t}{n}\right)^{n},\]
for any $t>0$, where $W_{(0)} \equiv 0$ and $W_{(n+1)} \equiv 1$. Therefore,
\begin{align*}
\min_{i\in\{0,\ldots,n\}} \left \{P_{(i+1)}-P_{(i)}\right\} \succeq O\left(\frac{1}{n^2(\log n)^{\sqrt{\epsilon/2}}}\right),
\end{align*}
with a large probability converging to 1, where $P_{(0)} \equiv 0$ and $P_{(n+1)} \equiv 1$.
Furthermore, by the following Bernstein inequality for binomial $X \sim \tbin(m,P)$,
\begin{align*}
pr(|X - m P|\geq t) \leq 2\exp \left\{-\frac{t^2}{2[mP(1-P)+t/3]}\right\},
\end{align*}
where $t>0$, and taking $t\sim O\left(\frac{m}{n^2(\log n)^{\sqrt{\epsilon/2}}}\right)$,
we have that
\begin{align*}
\min_{i,j\in\{0,\ldots,n\}} \left \{\frac{X_i}{m_i}-\frac{X_j}{m_j}\right\} \succeq O\left(\frac{1}{n^2(\log n)^{\sqrt{\epsilon/2}}}\right),
\end{align*}
with a large probability converging to 1 because $\lim_{n \rightarrow \infty} n^4(\log n)^{1+\epsilon}/ m_i = 0$ for any $m=m_i$ given by Assumption~1.
\end{proof}

\section{Lemma~\ref{lemma:nointersect}  and its proof}
\begin{lemma}\label{lemma:nointersect}
Given data which satisfy Equation~\eqref{eq:xm} of Lemma~\ref{lemma:nointersect0}, if $(\bU^\star,\bW^\star)$ are uniformly distributed on the set $\{(\bu^\star,\bw^\star): Q_{\bx}(\bu^\star,\bw^\star)\neq\emptyset\},$ we have that
\begin{equation}\label{eq:lemma2}
pr\{(G_i^*(x_i-1,U_i^\star), G_i^*(x_i,U_i^\star)]\cap (G_j^*(x_j-1,U_j^\star), G_j^*(x_j,U_j^\star)]\neq\emptyset \mbox{ ~for some~ $i\neq j$}\}\to 0.
\end{equation}
\end{lemma}

\begin{proof}
Let $\widetilde U_i$ be i.i.d. $U(0,1)$ and 
denote
$L_i=G_i^*(x_i-1,\widetilde U_i) \sim \tbeta(x_i,m_i-x_i+1)$, and $R_i= G_i^*(x_i,\widetilde U_i)\sim \tbeta(x_i+1,m_i-x_i)$.
Recall that the proposed Algorithm~1 can be regarded as an importance sampling in the following way: if there is one $k$-intersections (i.e., $(k+1)$ intervals share one common area) in $\{(L_i,R_i],i=1,\ldots,n\}$, the corresponding $W^*$ have $(k+1)!$ possible permutations. So we essentially need to show $\sum_{k=1}^{n-1} (k+1)!\times q_k \rightarrow 0$ as $n \to 0$, where $q_k$ is defined as the probability of $\{(L_i,R_i],i=1,\ldots,n\}$ having one or more $k$-intersections, and $W^*_{i_1}<\ldots<W^*_{i_{k+1}}$ (say $i_1,\ldots,i_{k+1}$ are indices of intervals which intersect).

We start with one-intersection between $i$-th and $j$-th intervals, i.e., $(L_i,R_i)\cap (L_j,R_j) \neq \emptyset$, which is equivalent to $L_1\leq R_2$ and $L_2\leq R_1$.
Without loss of generality, we focus on $L_1\leq R_2$ as $$pr(A_1\cap A_2)\leq \min\{pr(A_1),pr(A_2)\}.$$
For a random variable $Y\sim \tbeta(\alpha,\beta)$, let $\mu=E[Y]=\frac{\alpha}{\alpha+\beta}$.
By Theorem~2.1 of \cite{marchal2017}, $Y$ is sub-Gaussian with the variance proxy parameter
\begin{align*}
\Sigma\equiv  \frac{1}{4(\alpha+\beta+1)}.
\end{align*}
Hence by the definition of sub-Gaussian random variable, for any $c,t \in \mathbbm{R}$,
\begin{align*}
    pr\left(Y- \mu\geq t\right)\leq \exp\left\{-\frac{t^2}{2\Sigma}\right\}.
\end{align*}

Note that
\begin{align}\label{eq:p}
pr(L_1\leq R_2) = pr(R_2-L_1\geq 0),
\end{align}
where $L_1$ is sub-Gaussian with mean $x_1/(m_1+1)$ and $\Sigma_1=1/(m_1+2)$, and $R_2$ is sub-Gaussian with mean $(x_2+1)/(m_2+1)$ and $\Sigma_2=1/(m_2+2)$.
Therefore, Equation~\eqref{eq:p} equals $pr(Z\geq x_1/(m_1+1)-(x_2+1)/(m_2+1))$, where $Z$ is a sub-Gaussian with mean $0$ and $\Sigma_Z=1/(m_1+2)+1/(m_2+2)$.

Then we further have that Equation~\eqref{eq:p} equals
\begin{align*}
& pr\{Z\geq x_1/(m_1+1)-(x_2+1)/(m_2+1)\}\\
\leq  & \exp\left\{ - \frac{T^2}{2\Sigma_Z}  \right\}\\
\leq & \exp\left\{ - \tilde m T^2    \right\},
\end{align*}
where $\tilde m = \min_{i=1,\ldots,n} m_i$, and $T= x_1/(m_1+1)-(x_2+1)/(m_2+1)$.

Thus, the probability of a single intersection $2!\times q_1$ is bounded by
\begin{align*}
2! \times {n \choose 2} \times  \exp\left[ - \tilde m T^2    \right]
\preceq \exp\left\{c_1 \log n - c_2 (\log n)^{1+\epsilon/2}\right\},
\end{align*}
for any $\epsilon>0$, and some constants $c_1$ and $c_2$. We note that the coefficient $ {n \choose 2}$ refers to the number of possible pairs.

We next consider the case of $k$-intersections. We only need to consider two intervals corresponding to the two farthest $P_i$ among $k$ intervals.
 Thus, the probability of existing $k$-intersections $(k+1)!\times q_k$ is bounded by
\begin{align*}
 &(k+1)! \times {n \choose k} \times  \exp\left\{ - \tilde m \left[O\left(\frac{k}{n^2(\log n)^{\sqrt{\epsilon/2}}} \right)  \right]^2    \right\} \\
 \preceq & \exp\{ c_1(k\log n + k\log k) - c_2 k^2 (\log n)^{1+\epsilon/2} \},
\end{align*}
for any $\epsilon>0$, and some constants $c_1$ and $c_2$.
Because $\sum_{k=1}^{n-1} (k+1)!\times q_k \rightarrow 0$, we
conclude \eqref{eq:lemma2}.
\end{proof}

\section{Corollary~\ref{corollary} and its proof}
We present a Benstein-von Mises theorem for fiducial distribution associated with the empirical distribution function. This result can be viewed either as a special special case (without censoring) of \cite{cuihannig2019}, or as a particular case of exchangeably weighted bootstrap in \cite{praestgaard1993exchangeably}.

\begin{corollary}\label{corollary}
Assume the conditions of Theorem~\ref{main}. We have 
\[
n^{1/2}\{ \widetilde F(\cdot)-\widehat F_n(\cdot)  \} \rightarrow \{1-F(\cdot)\} B(\gamma(\cdot)),
\]
in distribution on Skorokhod space $\mathcal D[0,1]$ in probability, where $B$ is the Brownian Motion, $\gamma(t)=\int_0^t \frac{f(s)}{[1-F(s)]^2} ds= \frac{F(t)}{1-F(t)}$, $\widehat F_n$ is defined in Equation~(7) of Theorem~\ref{main}, and $\widetilde F$ is defined in Equation~(10) of Appendix~\ref{thm:3.1}.
\end{corollary}
\begin{proof}
By Theorem~2 of \cite{cuihannig2019}, we essentially need to check their Assumptions 1-3. Their Assumption~1 satisfies with their $\pi(p)= 1-F(p)$; their Assumption~2 satisfies as we assume true CDF is absolutely continuous; their Assumption~3 satisfies as $$\int_0^p \frac{g_n(s)}{\sum_{i=1}^n I(P_i\geq s)} d [\sum_{i=1}^n I(P_i\leq s)] \rightarrow \int_0^p \frac{f(s)}{[1-F(s)]^2} ds,$$ for any $p$ such that $1-F(p)>0$ and any sequence of functions $g_n \rightarrow \frac{1}{1-F}$ uniformly.
\end{proof}

\section{Remark on binomial and Poisson data}\label{sec:moments}

In the following two theorems for binomial and Poisson data respectively, we show implications of the fact that $pr(F \in Q_x(U^*,W^*))$ is proportional to the nonparametric likelihood function. In particular, maximizing the scaled fiducial probability in its limit provides exactly the underlying true CDF.

\begin{theorem}\label{main2}
Suppose $\Theta_i \equiv P_i$, and $X_i \mid P_i$ follows $\tbin(m,P_i)$, 
maximizing $\lim_{n \rightarrow \infty} [c_n \times pr(F\in Q_{\bx}(\bU^\star,\bW^\star))]^{1/n}$ leads to a CDF  matching the first $m$-moments of the true $F(p)$, where $c_n$ is the normalizing constant.
\end{theorem}
\begin{theorem}\label{remark:2}
Suppose $\Theta_i \equiv \Lambda_i$, and $X_i \mid \Lambda_i$ follows $\tpoi(\Lambda_i)$, 
maximizing $\lim_{n \rightarrow \infty} [c_n \times pr(F\in Q_{\bx}(\bU^\star,\bW^\star))]^{1/n}$ leads to true $F(\lambda)$ almost surely, where $c_n$ is the normalizing constant.
\end{theorem}

The details of proofs are provided below.

\subsection{Proof of Theorem~\ref{main2} \label{thm:2}}

\begin{proof}
Recall the data generating equation in \eqref{eq:DGEinApp}, 
where $G_i$ is the CDF of binomial distribution.
The fiducial probability is
\begin{align}
&pr(F\in Q_{\bx}(\bU^\star,\bW^\star))\nonumber \\  \propto & \prod_{i=1}^{n} \left[\int_0^1  {{m}\choose{x_i}} p^{x_i} (1-p)^{m-x_i} dF(p) \right] \nonumber \\
 =& \exp\left\{  \sum_{i=1}^{n}  \log \left( E_F \left[{{m}\choose{x_i}} p^{x_i} (1-p)^{m-x_i} \right]  \right)  \right\} \nonumber \\
 = &\exp\bigg \{n_{m} \log(E_F[P^m]) + n_{m-1} \log\left(E_F\left[ {{m}\choose{m-1}} P^{m-1}(1-P)\right] \right) + \cdots  \nonumber\\
& + n_{0} \log(E_F[(1-P)^m]) \bigg \}, \nonumber
\end{align}
where $n_{k}$ is the number of samples with $X_i=k$, and $E_F$ refers to the expectation with respect to $F$ that is evaluated. 
Thus, as $n$ goes to infinity, 
\begin{align}
& [c_n \times pr(F\in Q_{\bx}(\bU^\star,\bW^\star))]^{1/n} \nonumber \\
\rightarrow &\exp\bigg \{ E[P^m] \log(E_F[P^m]) + E\left[ {{m}\choose{m-1}} P^{m-1}(1-P)\right] \log\left(E_F\left[ {{m}\choose{m-1}} P^{m-1}(1-P)\right] \right)\nonumber\\&+ \cdots   + E[(1-P)^m] \log(E_F[(1-P)^m]) \bigg \},\label{eq:loglikelihood3}
\end{align}
where $c_n$ is the normalizing constant, and $E$ refers to the expectation with respect to the true distribution function that generates data. 
 By the method of Lagrange multipliers, \begin{align*}
 H(x) = \sum_i y_i \log x_i~~\text{subject to}~~\sum_i x_i=1,\sum_i y_i=1,
 \end{align*}
 is maximized with respect to $x$ by setting $x_k=y_k$, $k\in \mathbbm{R}^+$.
Maximizing Equation~\eqref{eq:loglikelihood3} gives
\begin{align*}
E[P^m]&=E_F[P^m], \\
E[P^{m-1}(1-P)]&=E_F[P^{m-1}(1-P)], \\
& \ldots \\
E[(1-P)^m]&=E_F[(1-P)^m].
\end{align*}
 The above equations are restrictions on the first $m$-moments of $P$, which completes the proof.
\end{proof}

\subsection{Proof of Theorem~\ref{remark:2} \label{thm:3}}
\begin{proof}
Recall the data generating equation is 
\begin{align*}
X_i= G_i^{-1}(U_i,\Lambda_i),
\quad 
\Lambda_i=F^{-1}(W_i),
\end{align*}
where $G_i$ is the CDF of Poisson distribution.
The fiducial probability is
\begin{align}
& pr(F\in Q_{\bx}(\bU^\star,\bW^\star)) \nonumber \\
\propto & \prod_{i=1}^n \int_0^\infty   \frac{ \lambda^{x_i} \exp\{-\lambda \}}{x_i!} dF(\lambda) \nonumber \\
=& \left\{\frac{E_F[\exp(-\Lambda)]}{0!}\right\}^{n_0}\times \left\{\frac{ E_F[\Lambda \exp(-\Lambda)]}{1!} \right \}^{n_1} \cdots \times \left\{\frac{ E_F[\Lambda^k \exp(-\Lambda)]}{k!} \right \}^{n_k}\times \cdots \nonumber\\
= & \exp\left \{n \left [ \frac{n_0}{n} \log \frac{E_F[\exp(-\Lambda)]}{0!}  +\frac{n_1}{n} \log \frac{E_F[\Lambda\exp(-\Lambda)]}{1!} + \cdots + \frac{n_k}{n} \log \frac{E_F[\Lambda^k\exp(-\Lambda)]}{k!} +\cdots \right]  \right\}, \nonumber
\end{align}
where $n_k$ is the count of $X_i=k$, and $E_F$ refers to the expectation with respect to $F$ that is evaluated.
 Thus, as $n$ goes to infinity, 
\begin{align}
& [c_n \times pr(F\in Q_{\bx}(\bU^\star,\bW^\star))]^{1/n} \nonumber \\
 \rightarrow & \exp\bigg \{ \bigg [ \frac{E[\exp(-\Lambda)]}{0!} \log \frac{E_F[\exp(-\Lambda)]}{0!}  + \frac{E[\Lambda \exp(-\Lambda)]}{1!} \log \frac{E_F[\Lambda\exp(-\Lambda)]}{1!} + \cdots \nonumber \\ & +  \frac{E[\Lambda^k \exp(-\Lambda)]}{k!} \log \frac{E_F[\Lambda^k\exp(-\Lambda)]}{k!} +\cdots \bigg]  \bigg\},
\label{eq:loglikelihood2} 
\end{align}
where $c_n$ is the normalizing constant, and $E$ refers to the expectation with respect to the true distribution function that generates data. 
 By the method of Lagrange multipliers, maximizing Equation~\eqref{eq:loglikelihood2} gives
\begin{align*}
\frac{E[\Lambda^k \exp(-\Lambda)]}{k!} =\frac{E_F[\Lambda^k \exp(-\Lambda)]}{k!}, ~ k \in \mathbbm{R}^+.
\end{align*}
 The above equations are essentially the restrictions on all derivatives of Laplace transform at 1.
The property of the Laplace transform being analytic in the region of absolute convergence implies
the uniqueness of a distribution,  
which completes the proof.
\end{proof}

\section{Additional simulation results with $n=1000$}

In this section, we present the results of both point estimates and 95\% CIs for $n=1000$. The simulations were replicated 500 times for each scenario. We again observe a consistent pattern that the proposed methods are comparable to and sometimes better than competing methods.


\begin{table}[!h]
\centering
\scalebox{1.1}{
\begin{tabular}{ccccccc}
  \hline
 Scenario & $p$ & F & g & bc & BP & BA \\
  \hline
      \multirow{2}{*}{1} & 0.15 & 0.50 & 1.79 & 1.74 & 10.65  & 0.04\\
& 0.25 & 1.49 & 1.46  &1.44 &52.03 &  0.65\\
& 0.50 &3.92  & 9.54  &9.52 &27.67  & 2.94\\
&0.75 & 1.32 &  0.61  & 0.62 & 5.48 &  0.67\\
& 0.85 & 0.48 &  0.17  & 0.16 & 1.30   & 0.04\\
   \hline
    \multirow{2}{*}{2} & 0.15 &  5  & 10 &   9  &122  &  2* \\
&0.25 & 4 &  8  &  8  & 19 &  14*\\
&0.50 &  3  &  3  &  3  &  3  &  4*\\
&0.75 & 4  &  3  &  3  & 17  & 16*\\
&0.95 & 4   & 1 &   1 &  18  &  2*\\
   \hline
   \multirow{2}{*}{3}  & 0.15 & 0.02 &  0.33  & 0.30   & 0.02  &  0.003\\
& 0.25  & 0.28 & 1.27  &1.21 &  0.65 &  0.10\\
& 0.50 & 3.02 & 6.98 & 6.95 & 2.81 & 2.21\\
&0.75  & 0.29 &  0.11 &  0.11 &  0.54 &  0.12\\
&0.85   & 0.02  & 0.02  &  0.01  &  0.02  &  0.002 \\
   \hline
 \multirow{2}{*}{4} & 0.15 & 3  &  5   & 5  &  3  &  2\\
& 0.25 &  3   & 2  &  2  &  3  &  2\\
&0.50 & 2  &  2  &  2 &   2  &  2\\
& 0.75 & 3  &  2  &  2  &  5  &  2\\
& 0.85 & 3 &  34  & 34 &   8   & 2\\
   \hline
  \multirow{2}{*}{5}  &   0.15 &  2  &  3  &  3   & 2  &  2*\\
     & 0.25 & 1  &  1  &  1  &  1  &  1* \\
     & 0.50 &0.20  &  0.26  &  0.25  &  0.20  &  0.17* \\
     & 0.75 & 0.02  & 0.11  & 0.10  &  0.02  &  0.02* \\
     & 0.85 & 0.01  &  0.04  &  0.04  &  0.01  &  0.01* \\
   \hline
\end{tabular}
}
\caption{MSE ($\times 10^{-4}$) of point estimates for $F(p)$ of each scenario.
``F'' denotes the fiducial point estimates;
``g'' denotes Efron's $g$-modeling without bias correction;  ``bc'' denotes Efron's $g$-modeling with bias correction;
``BP'' denotes the bootstrap method; ``BA'' denotes the Bayesian method.
*The Bayesian results for Scenarios 2, 5 are reported based on 495, 476 replications as 5, 24 runs failed due to an error in R package \texttt{dirichletprocess}.}
\end{table}

\begin{table}[!h]
\centering
\scalebox{1.1}{
\begin{tabular}{cccccccc}
  \hline
 Scenario & $p$ & M & C & g & bc & BP & BA \\
  \hline
      \multirow{2}{*}{1} &0.15 & 98 &  99 &  28  & 31  &  0  & 93\\
&0.25 & 100 & 100 &  91 &  92 &   0  & 83\\
&0.50 &  100 & 100 &  84 &  84&   10  & 70\\
&0.75 & 100 & 100 &  98  & 98  & 17 &  80\\
&0.85 & 98 &  99  & 98 &  98  & 10 &  88\\
   \hline
 \multirow{2}{*}{2}   & 0.15&   93 &  98 &  10  & 12  &  0 &  89*\\
& 0.25 &  100 & 100  & 82  & 82 &  17 &  87*\\
& 0.50 & 97  & 98 &  93  & 93  & 91  & 65*\\
&0.75 &  100  & 100  & 97  & 97  & 15 &  85*\\
& 0.85 & 95  & 97 &  98 &  98  &  0 &  89*\\
   \hline
 \multirow{2}{*}{3}   & 0.15&  99 &  99 &  74  & 80 &  72 &  31\\
& 0.25 &   98 &  99  & 68 &  71  & 77  & 29\\
& 0.50 & 99 & 100  & 71  & 71  & 94  & 20\\
&0.75 & 99  & 99  &100 &100  & 82  & 28\\
& 0.85 & 99 &  99 & 100  &100  & 69  & 28\\
   \hline
 \multirow{2}{*}{4}   & 0.15&  100 & 100 &  75 &  76  & 89  & 49\\
& 0.25 &  95  & 95  & 96 &  96 &  91  & 13 \\
& 0.50 & 96  & 96 &  96 &  96 &  96 &   3\\
&0.75 &94  & 96 &  96 &  96  & 85 &  12\\
& 0.85 & 100 & 100 &   2  &  2  & 55  & 53\\
   \hline
  \multirow{2}{*}{5}  &   0.15 &  99 &  100 &  87 &  87  & 93 &  22* \\
     & 0.25 & 99  & 99 &  96 &  96 &  96 &  21* \\
     & 0.50 & 98  & 98  & 98 &  98  & 94  & 29* \\
     & 0.75 & 99 & 100  & 99  & 99  & 90 &  29* \\
     & 0.85 & 99 &  99  & 100 & 100  & 56  & 26* \\
   \hline
\end{tabular}
}
\caption{Coverage (in percent) of 95\% CIs for $F(p)$ of each scenario. ``M'' denotes mixture GFD confidence intervals;  ``C'' denotes conservative GFD confidence intervals;
``g'' denotes Efron's $g$-modeling without bias correction;  ``bc'' denotes Efron's $g$-modeling with bias correction;
``BP'' denotes the bootstrap method; ``BA'' denotes the Bayesian method.
*The Bayesian results for Scenarios 2, 5 are reported based on 495, 476 replications as 5, 24 runs failed due to an error in R package \texttt{dirichletprocess}.}
\end{table}

\begin{table}[!h]
\centering
\scalebox{1.1}{
\begin{tabular}{cccccccc}
  \hline
 Scenario & $p$ & M & C & g & bc & BP & BA \\
  \hline
      \multirow{2}{*}{1} & 0.15 & 27  & 29  & 23 &  23 &  23 &   8\\
&  0.25 &  64  & 69 &  39 &  39 &  40 &  26\\
& 0.50 & 147 & 156  & 86 &  86   &62 &  47\\
&0.75 & 64  & 69 &  35   &35  & 32  & 26\\
& 0.85 & 26 &  29   &15  & 15  & 16 &  8\\
   \hline
    \multirow{2}{*}{2} &0.15 & 69  & 75 &  43 &  43 &  43 &  45*\\
& 0.25 & 132 & 142 &  74 &  74  & 57 & 120*\\
&0.50 & 72 &  76 &  65  & 65 &  62  & 35*\\
& 0.75 & 132 & 142  & 72 &  72 &  51&  121*\\
& 0.85 & 70 &  75  & 35 &  35 &  32  & 45*\\
   \hline
   \multirow{2}{*}{3}  & 0.15 &  7  &  7  & 12 &  12  &  3 &   0.49\\
& 0.25 &  25 &  27  & 24  & 24  & 19  &  3\\
& 0.50 & 92  & 96  & 60 &  60 &  62  &  8\\
& 0.75  &  25  & 27  & 17 &  17  & 18  &  3\\
 & 0.85 & 7  &  7  &  4 &   4  &  3  &  0.49\\
   \hline
 \multirow{2}{*}{4}  &0.15&  99 & 105 &  56  & 56 &  57  & 24\\
& 0.25 & 65 & 67 &  62  & 62  & 62  &  6\\
&0.50&  62 &  63  & 62  & 62 &  62 &   0.02 \\
&0.75&  65 &  67 &  61  & 61  & 62  &  5\\
&0.85& 99  & 105   & 57 &  57  & 55  & 24\\
   \hline
  \multirow{2}{*}{5}  &   0.15 & 81 &  86 &  52  & 52 &  56  &  8* \\
     & 0.25 & 57 &  60 &  44 &  44  & 43  &  6* \\
     & 0.50 & 21 &  22&   20 &  20  & 17  &  3*\\
     & 0.75 &  7 &   8 &  11 &  11  &  5  &  1* \\
     & 0.85 &  5  &  6 &   6  &  6   & 2   & 1* \\
   \hline
\end{tabular}
}
\caption{Mean length ($\times 10^{-3}$) of 95\% CIs for $F(p)$ of each scenario.  ``M'' denotes mixture GFD confidence intervals;  ``C'' denotes conservative GFD confidence intervals;
``g'' denotes Efron's $g$-modeling without bias correction;  ``bc'' denotes Efron's $g$-modeling with bias correction;
``BP'' denotes the bootstrap method; ``BA'' denotes the Bayesian method.
*The Bayesian results for Scenarios 2, 5 are reported based on 495, 476 replications as 5, 24 runs failed due to an error in R package \texttt{dirichletprocess}.}
\end{table}

\section{Trace plots of the proposed Gibbs sampler for intestinal surgery data}
In this section, we present the trace plots of the proposed Gibbs sampler for intestinal surgery data.
As can be seen from these figures, the fiducial MCMC samples have good variability and mix well.

\begin{figure}[H]
\begin{center}
\includegraphics[width=12cm]{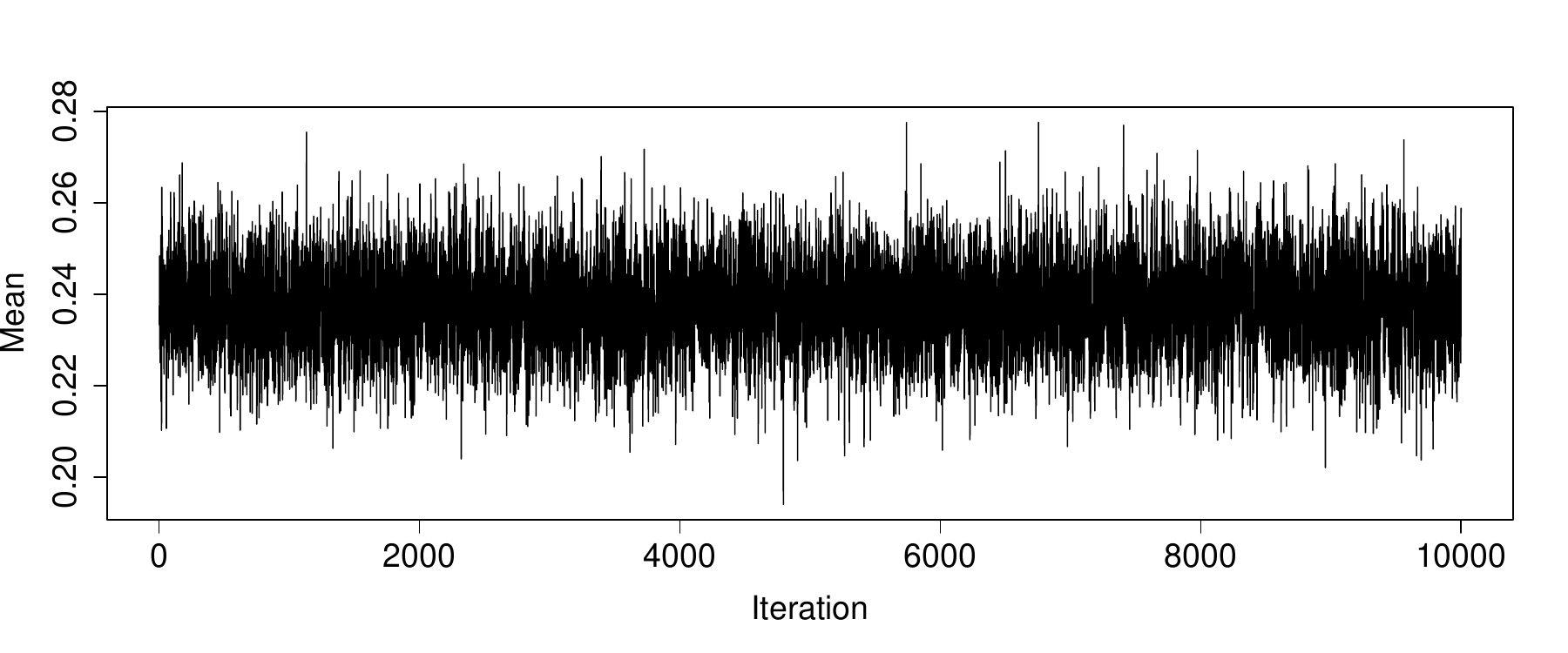}
\includegraphics[width=12cm]{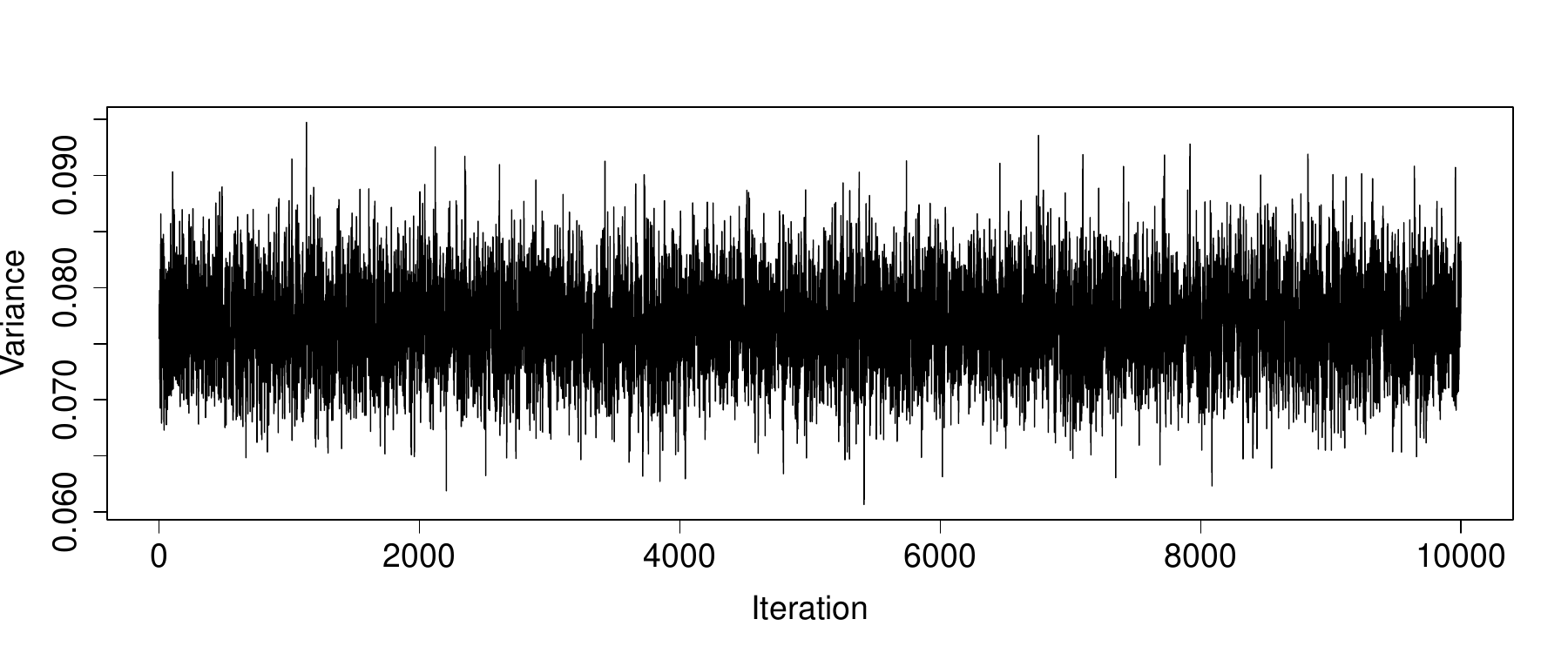}
\end{center}
\caption{Trace plots of mean and variance of $P\sim [F^L(p)+F^U(p)]/2$ for intestinal surgery data, respectively.}
\end{figure}


\section{Additional plots for intestinal surgery data}
In Figure~\ref{pic:real2}, we plot the point estimates and 95\% confidence intervals of $F(p)$ for the Bayesian method and nonparametric bootstrap, respectively.

\begin{figure}[H]
\begin{center}
\includegraphics[width=5.5cm]{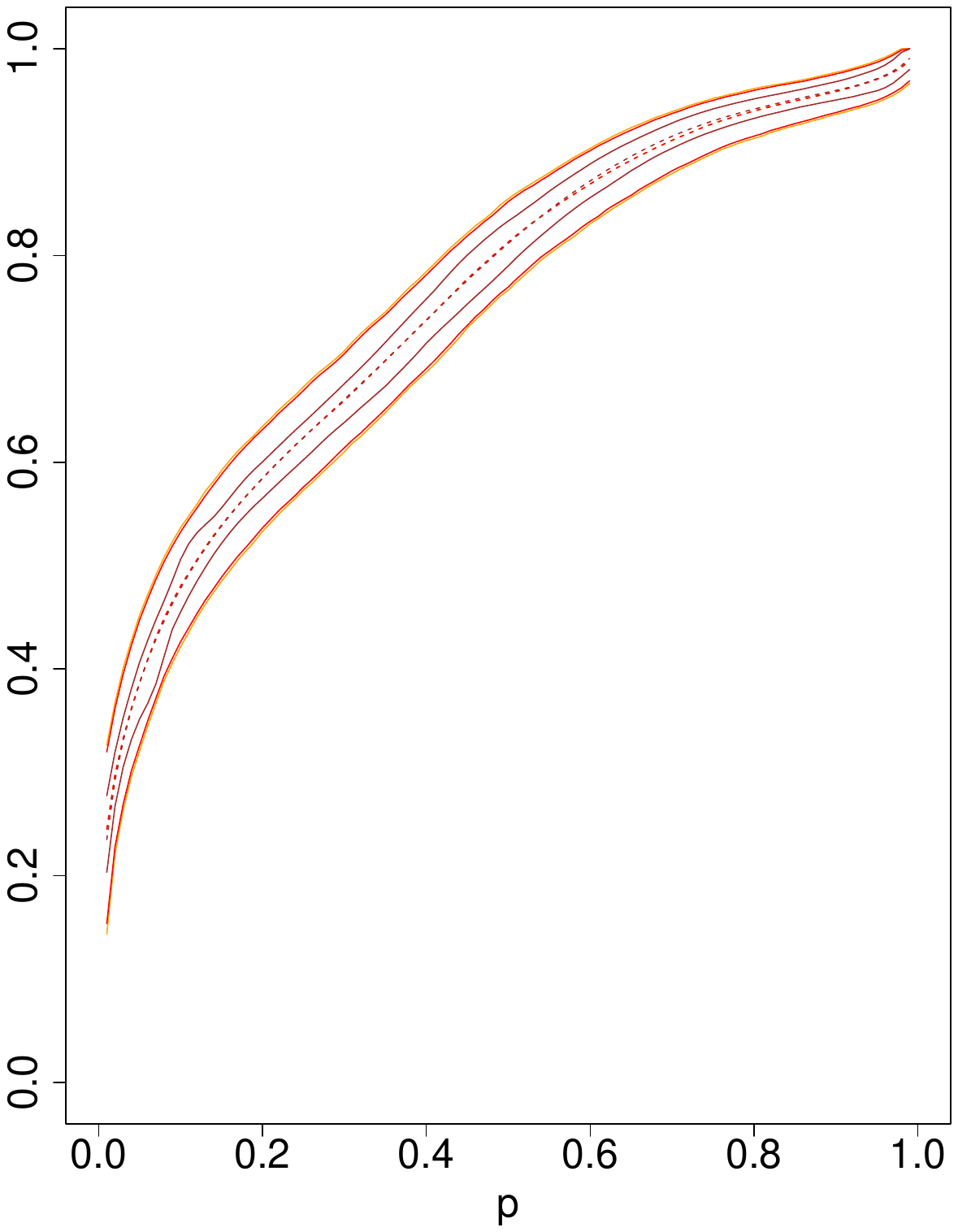}
\includegraphics[width=5.5cm]{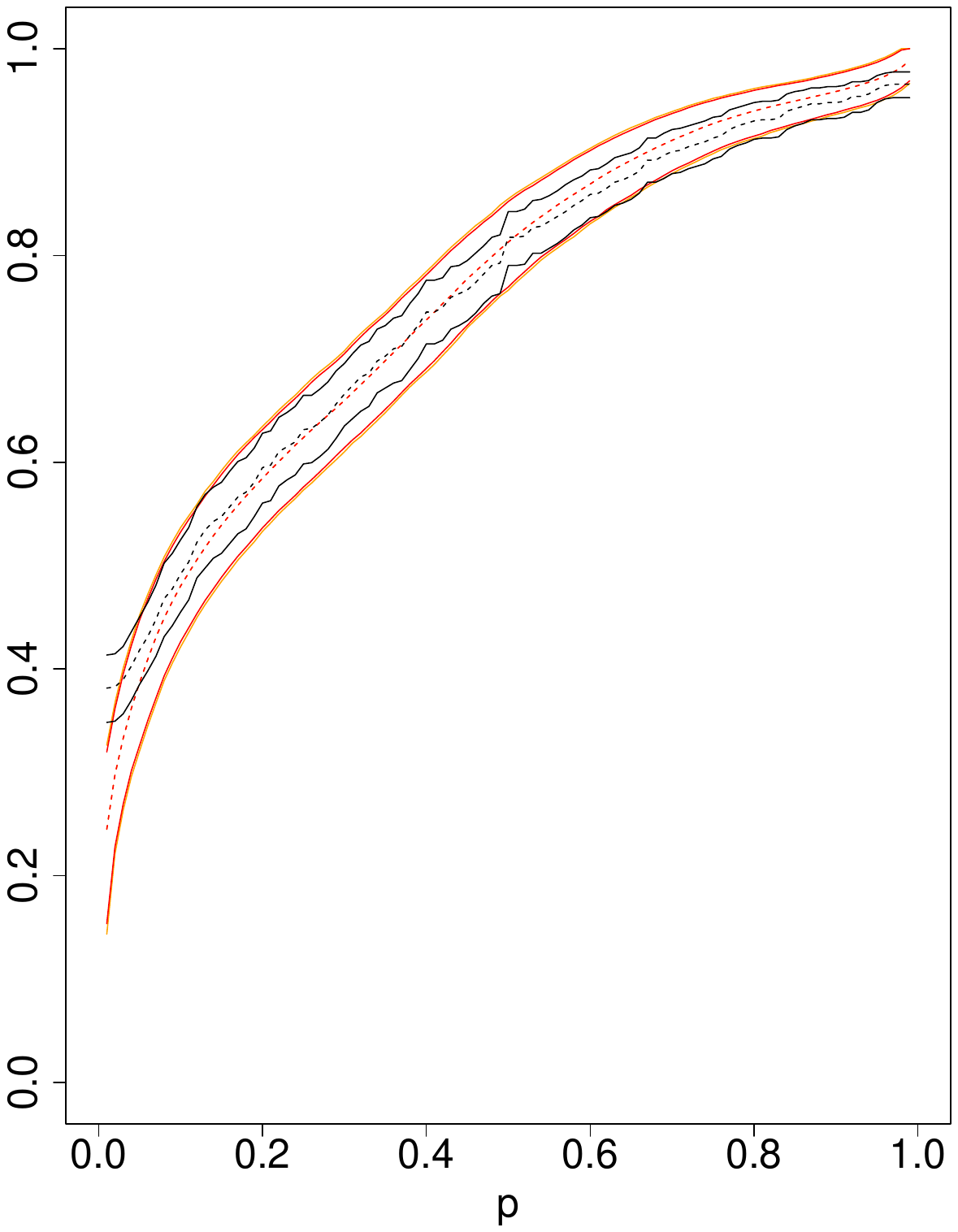}
\end{center}
\caption{Estimated CDF (dashed line) and 95\% CIs for $F(p)$. Left panel: fiducial versus Bayesian.
Right panel: fiducial versus nonparametric bootstrap. The red and orange curves are mixture and conservative confidence intervals, respectively.
The brown and black curves are Bayesian and bootstrap confidence intervals, respectively.
}
\label{pic:real2}
\end{figure}

\bibliographystyle{asa}
\bibliography{fiducial}

\end{document}